\documentclass[fullpage,11pt]{article}

\usepackage{amsthm}
\usepackage{graphicx} 
\usepackage{array} 

\usepackage{amsmath, amssymb, amsfonts, verbatim}
\usepackage{hyphenat,epsfig,subfigure,multirow}

\usepackage[usenames,dvipsnames]{xcolor}
\usepackage[ruled]{algorithm2e}

\usepackage{tcolorbox}
 
\definecolor{DarkRed}{rgb}{0.5,0.1,0.1}
\definecolor{DarkBlue}{rgb}{0.1,0.1,0.5}

\usepackage{hyperref}
\hypersetup{
colorlinks=true,
pdfnewwindow=true,
citecolor=ForestGreen,
linkcolor=DarkRed,
filecolor=DarkRed,
urlcolor=DarkBlue
}

\usepackage{bm}
\usepackage{url}
\usepackage{xspace} 
\usepackage[mathscr]{euscript}

\usepackage{mdframed}

\usepackage[noend]{algpseudocode}
\makeatletter
\def\BState{\State\hskip-\ALG@thistlm}
\makeatother

\usepackage{cite}
\usepackage{enumerate}

\usepackage[margin=1in]{geometry}

\newtheorem{theorem}{Theorem}
\newtheorem{lemma}{Lemma}[section]
\newtheorem{proposition}[lemma]{Proposition}

\newtheorem{claim}[lemma]{Claim}
\newtheorem{fact}[lemma]{Fact}

\newtheorem{definition}{Definition}
\newtheorem{problem}{Problem}
\newtheorem{remark}[lemma]{Remark}
\newtheorem*{claim*}{Claim}
\newtheorem*{proposition*}{Proposition}
\newtheorem*{lemma*}{Lemma}
\newtheorem*{problem*}{Problem}

\newtheorem{mdresult}{Result}
\newenvironment{result}{\begin{mdframed}[backgroundcolor=lightgray!40,topline=false,rightline=false,leftline=false,bottomline=false,innertopmargin=2pt]\begin{mdresult}}{\end{mdresult}\end{mdframed}}

\allowdisplaybreaks

\renewcommand{\qed}{\nobreak \ifvmode \relax \else
      \ifdim\lastskip<1.5em \hskip-\lastskip
      \hskip1.5em plus0em minus0.5em \fi \nobreak
      \vrule height0.75em width0.5em depth0.25em\fi}

\usepackage[T1]{fontenc}
\usepackage[utf8]{inputenc}

\newcommand{\ourinfo}{Supported in part by National Science Foundation grants CCF-1552909, CCF-1617851, and IIS-1447470.}

\newcommand{\toShrink}{-.20cm}
\newcommand{\toShrinkEnu}{-.2cm}

\newcommand{\itfacts}[1]{Fact~\ref{fact:it-facts}-(\ref{part:#1})\xspace}

\newcommand{\algline}{
  \rule{0.5\linewidth}{.1pt}\hspace{\fill}%
  \par\nointerlineskip \vspace{.1pt}
}

\newcommand{\Ot}{\ensuremath{\widetilde{O}}}
\newcommand{\eps}{\ensuremath{\varepsilon}}

\newcommand{\Bracket}[1]{\Big[#1\Big]}
\newcommand{\bracket}[1]{\left[#1\right]}
\newcommand{\paren}[1]{\ensuremath{\left(#1\right)}\xspace}
\newcommand{\card}[1]{\left\vert{#1}\right\vert}
\newcommand{\Omgt}{\ensuremath{\widetilde{\Omega}}}

\newcommand{\norm}[1]{\ensuremath{\|#1\|}}

\newcommand{\set}[1]{\ensuremath{\left\{ #1 \right\}}}
\newcommand{\poly}{\mbox{\rm poly}}
\newcommand{\polylog}{\mbox{\rm  polylog}}

\newcommand{\OPT}{\ensuremath{\mbox{\sc opt}}\xspace}
\newcommand{\opt}{\textnormal{\ensuremath{\mbox{opt}}}\xspace}

\newcommand{\alg}{\ensuremath{\mathcal{A}}\xspace}

\DeclareMathOperator*{\Exp}{\ensuremath{{\mathbb{E}}}}

\renewcommand{\Pr}{\Prob}
\newcommand{\EX}{\Exp}
\newcommand{\Ex}{\Exp}

\newcommand{\etal}{et al.\xspace}

\newcommand{\event}[1]{\ensuremath{{\sf E}_{#1}}}

\newenvironment{tbox}{\begin{tcolorbox}[
		enlarge top by=5pt,
		enlarge bottom by=5pt,
		 boxsep=0pt,
                  left=4pt,
                  right=4pt,
                  top=10pt,
                  arc=0pt,
                  boxrule=1pt,toprule=1pt,
                  colback=white
                  ]
	}
{\end{tcolorbox}}

\newcommand{\Yes}{\ensuremath{\textnormal{\textsf{Yes}}}\xspace}
\newcommand{\No}{\ensuremath{\textnormal{\textsf{No}}}\xspace}

\newcommand{\dist}{\ensuremath{\mathcal{D}}}

\newcommand{\Prot}{\ensuremath{\Pi}}
\newcommand{\prot}{\ensuremath{\pi}}

\newcommand{\bA}{\bm{A}}

\newcommand{\bB}{\ensuremath{\bm{B}}}
\newcommand{\bC}{\ensuremath{\bm{C}}}
\newcommand{\bD}{\ensuremath{\bm{D}}}

\newcommand{\supp}[1]{\ensuremath{\textsc{supp}(#1)}}

\newcommand{\ICost}[2]{\ensuremath{\textnormal{\textsf{ICost}}_{#2}(#1)}\xspace}
\newcommand{\IC}[3]{\ensuremath{\textnormal{\textsf{IC}}_{#2}^{#3}(#1)}\xspace}
\newcommand{\CC}[3]{\ensuremath{\textnormal{\textsf{CC}}_{#2}^{#3}(#1)}\xspace}

\renewcommand{\event}{\mathcal{E}}

\newcommand{\istar}{\ensuremath{i^{\star}}}

\newcommand{\textbox}[2]{
{
\begin{tbox}
\textbf{#1}
{#2}
\end{tbox}
}
}

\newcommand{\topt}{\ensuremath{\widetilde{\textnormal{opt}}}}

\renewcommand{\event}[1]{\mathcal{E}(#1)}

\renewcommand{\OPT}{\ensuremath{\textnormal{\textsf{OPT}}}}

\renewcommand{\event}[1]{\ensuremath{\mathcal{E}\paren{#1}}}

\renewcommand{\bA}{\ensuremath{\overline{A}}}
\renewcommand{\bB}{\ensuremath{\overline{B}}}

\newcommand{\FC}{\ensuremath{\mathcal{F}}}

\renewcommand{\event}{\ensuremath{\mathcal{E}}\xspace}

\newcommand{\estar}{\ensuremath{e^{\star}}}

\newcommand{\Disj}{\ensuremath{\textnormal{\textsf{Disj}}}\xspace}
\newcommand{\Disjt}{\ensuremath{\textnormal{\textsf{Disj}}_{t}}\xspace}
\newcommand{\SetCover}{\ensuremath{\textnormal{\textsf{SetCover}}}\xspace}

\newcommand{\distDisj}{\ensuremath{\dist_{\Disj}}\xspace}
\newcommand{\distDisjY}{\ensuremath{\dist^{\textsf{Y}}_{\Disj}}\xspace}
\newcommand{\distDisjN}{\ensuremath{\dist^{\textsf{N}}_{\Disj}}\xspace}

\newcommand{\distSC}{\ensuremath{\dist_{\textnormal{\textsf{SC}}}}\xspace}

\renewcommand{\bar}[1]{\ensuremath{\overline{#1}}}

\renewcommand{\SS}{\ensuremath{\mathcal{S}}}
\newcommand{\ST}{\ensuremath{\mathcal{T}}}
\newcommand{\SC}{\ensuremath{\mathcal{C}}}
\newcommand{\SZ}{\ensuremath{\mathcal{Z}}}

\renewcommand{\Ex}[1]{\ensuremath{\mathbb{E}\bracket{#1}}}
\renewcommand{\Pr}[1]{\ensuremath{\mathbb{P}\paren{#1}}}

\newcommand{\protSC}{\ensuremath{\prot_{\textsf{SC}}}\xspace}
\newcommand{\protDisj}{\ensuremath{\prot_{\Disj}}\xspace}

\newcommand{\II}{\ensuremath{\mathbb{I}}}
\newcommand{\HH}{\ensuremath{\mathbb{H}}}

\renewcommand{\bA}{\ensuremath{\bm{A}}}
\renewcommand{\bB}{\ensuremath{\bm{B}}}

\newcommand{\PR}{\ensuremath{\mathbb{P}}}

\newcommand{\errs}{\ensuremath{\textnormal{errs}}\xspace}

\newcommand{\distN}{\ensuremath{\dist^{\textsf{N}}}}

\renewcommand{\prot}{\Pi}

\title{Tight Space-Approximation Tradeoff for the Multi-Pass Streaming Set Cover Problem}
\author{Sepehr Assadi\thanks{\ourinfo}\\University of Pennsylvania \\ {sassadi@cis.upenn.edu}}

\date{}

\begin{document}
\maketitle

\thispagestyle{empty}
\begin{abstract}
	We study the classic set cover problem in the streaming model: the sets that comprise the instance are revealed one by one in a stream and the goal is to solve
	the problem by making one or few passes over the stream while maintaining a sublinear space $o(mn)$ in the input size; here $m$ denotes the number of the sets and $n$ is the universe size.
	Notice that in this model, we are mainly concerned with the space requirement of the algorithms and hence do not restrict their computation time.
	
	Our main result is a resolution of the space-approximation tradeoff for the streaming set cover problem: we show that any $\alpha$-approximation algorithm
	for the set cover problem requires $\Omgt(mn^{1/\alpha})$ space, even if it is allowed $\polylog{(n)}$ passes over the stream, and even if the sets are arriving in a random order in the stream.  
	This space-approximation tradeoff matches the best known bounds achieved by the recent algorithm of Har-Peled~\etal (PODS 2016) that requires
	only $O(\alpha)$ passes over the stream in an adversarial order, hence settling the space complexity of approximating the set cover problem in data streams in a quite robust manner.  
	Additionally, our approach yields tight lower bounds for the space complexity of $(1- \eps)$-approximating the streaming maximum coverage problem studied in several recent works. 
\end{abstract}
\clearpage
\setcounter{page}{1}

\newcommand{\PC}{\ensuremath{\textnormal{\textsf{P}}}}
\newcommand{\NPC}{\ensuremath{\textnormal{\textsf{NP}}}}

\section{Introduction}\label{sec:intro}

The \emph{set cover} problem is one of the most fundamental optimization problems in computer science, with a wide range of 
applications in various domains including data mining and information retrieval~\cite{Anagnostopoulos15,SahaG09}, web host analysis~\cite{ChierichettiKT10}, operation 
research~\cite{GrossmanA97}, and many others.  
In this problem, we are given a collection of $m$ sets from a universe $[n]$ and the goal is to output a smallest number of sets whose union is $[n]$, or in other words, \emph{cover} the universe. 
The set cover problem is one of Karp's original 21 \NPC-hard problems~\cite{Karp72}. A simple greedy algorithm that iteratively picks the set that covers the most number of uncovered elements
achieves a $(\ln{n})$-approximation~\cite{Johnson74a,Slavik97} and this is best possible unless $\PC = \NPC$~\cite{DinurS14,Feige98,LundY94,Moshkovitz15}. 

The aforementioned results focus on the tradeoff between approximation guarantee and time complexity of the set cover problem. 
Nevertheless, in many settings, \emph{space complexity} of the algorithms is crucial to optimize. A canonical example is in applications in big data analysis: in such settings, one would like
to design algorithms capable of processing massive datasets using only few passes over the input and limited space. 
The well-established \emph{streaming model} of computation~\cite{AlonMS96,Muth05} precisely captures this setting. 

In the \emph{streaming set cover} problem, originally introduced by
Saha and Getoor~\cite{SahaG09}, the input sets are provided one by one in a stream and the algorithms are allowed to make a small number of passes
over the stream while maintaining a sublinear space $o(mn)$ for processing the stream. The streaming set cover problem and the closely related maximum coverage problem have
received quite a lot of attention in recent years~\cite{SahaG09,CormodeKW10,EmekR14,DemaineIMV14,BadanidiyuruMKK14,IndykMV15,HarPeledIMV16,ChakrabartiW16,AssadiKL16,McGregorVu16,BateniEM16}; 
we refer the reader to~\cite{AssadiKL16,McGregorVu16} for a comprehensive summary of these results.

Particularly relevant to our work, Demaine~\etal~\cite{DemaineIMV14}, have shown an $\alpha$-approximation algorithm that uses $O(\alpha)$ passes over the stream and needs $\Ot(mn^{\Theta(1/\log{\alpha})})$ space. 
Recently, Har-Peled~\etal~\cite{HarPeledIMV16} provide a significant improvement over this algorithm: they developed an 
$\alpha$-approximation, $O(\alpha)$-pass streaming algorithm that requires $\Ot(mn^{\Theta(1/\alpha)})$ space. They further conjectured that the tradeoff between the number of 
passes and the space in their algorithm is almost tight: this is supported by a lower bound of $\Omgt(mn^{1/2p})$ space for $p$-pass streaming algorithms that
compute an \emph{exact} set cover solution~\cite{HarPeledIMV16}. 

Notice however that the algorithm of~\cite{HarPeledIMV16} (and~\cite{DemaineIMV14}) exhibits a somewhat unusual behavior: allowing a larger number of passes 
over the stream results in a weaker approximation guarantee obtained by the algorithm. This highlights the following natural question: can we achieve
a (fixed) \emph{constant} approximation in $p$-passes and $\Ot(mn^{\Theta(1/p)})$ space? (a recent algorithm of Bateni~\etal~\cite{BateniEM16} achieves a fixed $\log{n}$-approximation within these bounds.)   
In general, what is the \emph{space-approximation tradeoff} for the streaming set cover problem if we consider algorithms that are allowed a relatively small number of passes, say up to $\polylog{(n)}$, over the stream? This is precisely the question addressed in this work. 

\subsection{Our Contributions}\label{sec:results}

Our main result is a tight resolution of the space-approximation tradeoff for the streaming set cover problem: 

\begin{result}[Main result, formalized as Theorem~\ref{thm:sc-lower}]\label{res:sc-lower}
	 Any streaming $\alpha$-approximation $\polylog{(n)}$-pass algorithm for the set cover problem requires $\Omgt(mn^{{1}/{\alpha}})$ space even on {random arrival streams}.
	 This lower bound applies even for the weaker goal of estimating the optimal value of the set cover instance (as opposed to finding the actual sets that cover the universe). 
\end{result}

Prior to our work, the best known lower bounds for \emph{randomized multi-pass} streaming algorithms ruled out the possibility of $(\log{n}/2)$-approximation in $p$ passes and $o(m/p)$ space~\cite{Nisan02}, and
exact solution in $p$ passes and $o(mn^{1/2p})$ space~\cite{HarPeledIMV16} (the later holds only if $m = O(n)$). These results left open the possibility of obtaining, say, a $2$-approximation in 
two passes or even an exact answer in $O(\log{n})$ passes and $\Ot(m)$ space. On the other hand, Result~\ref{res:sc-lower} smoothly extends the bounds in~\cite{Nisan02} to the whole range of 
approximation factors $\alpha = o(\log{n})$, proving the first \emph{super-linear} in $m$ lower bound for approximating set cover in \emph{multi-pass} streams. It also significantly improves the 
bounds in~\cite{HarPeledIMV16} to $\Omgt(mn/p)$ (and all range of $m = \poly(n)$) for $p$ pass streaming algorithms that recover an exact answer\footnote{Note that this result also implies that the ``right'' tradeoff between space 
and number of passes for obtaining an \emph{exact} solution to the streaming set cover is in fact linear as opposed to exponential, i.e., $n/p$ as opposed to $n^{1/p}$, as was previously shown in~\cite{HarPeledIMV16}.}. 

As mentioned earlier, Har-Peled~\etal~\cite{HarPeledIMV16} designed an $\alpha$-approximation algorithm for the set cover problem that requires $\Ot(mn^{\Theta(1/\alpha)})$ space 
(for some unspecified constant larger than $2$ in the $\Theta$-notation in the exponent). We can show that with proper modifications, 
this algorithm in fact only requires $\Ot(mn^{1/\alpha})$ space (see Theorem~\ref{thm:sc-upper}), hence proving a \emph{tight} upper bound for Result~\ref{res:sc-lower} (up to logarithmic factors). 
These results together resolve the space-approximation tradeoff for streaming set cover
problem in multi-pass streams. It is worth mentioning that the space-approximation tradeoff for \emph{single-pass} streaming algorithms
 of set cover has been previously resolved in~\cite{AssadiKL16}. 

Finally, we point out that the lower bound in Result~\ref{res:sc-lower} is quite \emph{robust} in the sense that it holds even when the sets are arriving in a random order. 
This is particularly relevant to the streaming set cover problem as most known techniques for this problem are based on element and set sampling and a-priori one may expect
that random arrival streams can facilitate the use of such techniques, resulting in better bounds than the ones achievable in adversarial streams. We point that in general, 
many streaming problems are known to be distinctly easier in random arrival streams compared to adversarial streams (see, e.g.,~\cite{GuhaM09,KonradMM12,KapralovKS14}). 

We further show an application of our techniques in establishing Result~\ref{res:sc-lower} to the \emph{streaming maximum coverage} problem
 that has been studied in several recent works~\cite{SahaG09,AusielloBGLP12,BadanidiyuruMKK14,McGregorVu16,BateniEM16,EpastoLVZ16,ChenNZ16}. 
In this problem, we are given a collection of $m$ sets from a universe $[n]$ and an integer $k \geq 1$, and the goal is to find $k$ sets that cover the most number of elements in $[n]$. We prove that,

\begin{result}[Formalized as Theorem~\ref{thm:kc-lower}]\label{res:kc-lower}
 	Any streaming $(1-\eps)$-approximation $\polylog{(n)}$-pass algorithm for the maximum coverage problem requires $\Omgt(m/\eps^2)$ space even on {random arrival streams}. This lower 
	bound applies even for the case $k = O(1)$. 
\end{result}

Single-pass $(1-\eps)$-approximation algorithms for this problem that use, respectively, $\Ot(mk/\eps^2)$ space and $\Ot(m/\eps^3)$ have been proposed recently 
in~\cite{McGregorVu16,BateniEM16}, and~\cite{BateniEM16}. Our Result~\ref{res:kc-lower} is hence \emph{tight} for any $k = O(1)$ (up to logarithmic factors) and 
within an $O(1/\eps)$ factor of the best upper bound for the larger values of $k$. 

McGregor and Vu~\cite{McGregorVu16} have very recently proved an $\Omgt(m)$ lower bound for $\polylog{(n)}$-pass streaming algorithms that approximate
the maximum coverage problem to within a factor better than $(1-1/e)$ (a single-pass $(1-1/e)$-approximation algorithm in $\Ot(m)$ space is also developed in~\cite{McGregorVu16,BateniEM16}).
The importance of Result~\ref{res:kc-lower} is thus in establishing the tight dependence on the parameter $\eps$ for this problem. This is important as $(1-\eps)$-approximation algorithms
for this problem for very small values of $\eps$, i.e., $\eps = 1/n^{\Omega(1)}$, are typically used as a sub-routine in approximating the streaming set cover problem in
multiple passes~\cite{DemaineIMV14,HarPeledIMV16,BateniEM16} (see Section~\ref{sec:sc-upper} for more details). 

En route, we also obtain the following result which may be of independent interest: the communication complexity of computing an exact solution to the set cover problem or the maximum coverage problem in 
the two-player communication model is $\Omgt(mn)$ bits (see Theorems~\ref{thm:sc-cc-lower} and~\ref{thm:kc-cc-lower}). This improves upon the previous 
$\Omega(m)$ lower bounds of Nisan~\cite{Nisan02} (for set cover) and McGregor and Vu~\cite{McGregorVu16} (for maximum coverage). The two-player communication model 
for set cover has also been studied in~\cite{AssadiKL16,ChakrabartiW16,DemaineIMV14,HarPeledIMV16}.  

We conclude this section by highlighting the following important aspect of our lower bounds. 
\begin{remark}\label{rem:np-hardness} 
In the hard instances we consider in proving Results~\ref{res:sc-lower} and~\ref{res:kc-lower}, the minimum set cover size and the parameter $k$ in maximum
coverage are \emph{small constants} and hence these instances admit a trivial poly-time algorithm in the classical (offline) setting. Our results hence establish the ``hardness'' of these instances
under the space restrictions of the streaming model, independent of the \NPC-hardness of approximating these problems.  
\end{remark}



\subsection{Technical Overview}\label{sec:techniques} 

We focus here on providing a technical overview of the proof of Result~\ref{res:sc-lower} - Result~\ref{res:kc-lower} is also proven
along similar lines. The starting point of our work is~\cite{AssadiKL16}, which proved a tight space lower bound for single-pass streaming algorithms
of set cover by analyzing the \emph{one-way communication complexity} of this problem (see Section~\ref{SEC:PRELIM} for details on communication complexity).

The overall approach of~\cite{AssadiKL16} can be summarized as follows. Consider a communication problem whereby
 Alice is given a collection of sets $S_1,\ldots,S_m$, Bob is given a set $T$, and they need to compute an $\alpha$-approximation
 of the set cover instance $(S_1,\ldots,S_m,T)$ in the one-way communication model. 
The input to the players are correlated in that there exists a set $S_{\istar}$ in Alice's collection which together with Bob's set $T$ cover the whole universe except for a single element. 
However, if the content of the set $S_{\istar}$ is unknown to Bob, i.e., Alice's message does not reveal almost all $S_{\istar}$, 
Bob needs to cover $[n] \setminus T$ (which is a subset of $S_{\istar}$ except for one element) with sets other than $S_{\istar}$ to ensure that the single 
element outside $S_{\istar}$ is covered. The collection $S_1,\ldots,S_m$ is designed to satisfy the so-called \emph{$r$-covering} property~\cite{LundY94} that states that no small collection of $S_i$'s set
can cover another set $S_j$ entirely\footnote{It is worth mentioning that essentially all known lower bounds for the streaming set cover problem, on their core, are based on some variant of this $r$-covering property; 
see~\cite{ChakrabartiW16} for more details.}, hence forcing  Bob to use many sets to cover the universe. The authors then use the \emph{information complexity} paradigm to reduce
the set cover problem on this distribution to multiple instances of a simpler problem (called the \emph{Trap} problem) and prove a lower bound for this new problem. 

In this paper, we extend this approach to lower bound the two-way communication complexity of the set cover problem and ultimately obtain the desired lower bound in Result~\ref{res:sc-lower} 
for multi-pass streaming algorithms. To do this, we need to address the following issues: 

First, the type of distribution used in~\cite{AssadiKL16} is clearly not suitable for proving lower bounds in the two-way model. In particular, we need a distribution 
with both Alice and Bob having $\Omega(m)$ sets and additionally, no clear ``signal'' to either party as which of the sets are more important, i.e., correspond to the sets $S_{\istar}$ and $T$ in the above distribution. To achieve this, 
we employ the $r$-covering property in a novel way: we first design a collection of sets $Z_1,\ldots,Z_m$ such that no collection of $\alpha$ sets $Z_i$'s can cover the universe $[n]$ unless they contain a  
single set $Z_{\istar}$ which is in fact equal to $[n]$ already (for remaining sets $Z_i$, we have $\card{Z_i} \approx n-n^{1-1/\alpha}$). 
Next, we decompose each $Z_i$ into two sets $S_i$ and $T_i$ and provide Alice with $S_i$, and Bob with $T_i$. This way, the sets $S_{\istar}$ and $T_{\istar}$ form a 
set cover of size two, and the $r$-covering property ensures that no other collection of $\alpha$ pairs $(S_i,T_i)$ can cover the universe; we further prove that ``mix and matching'' the sets (i.e., picking $S_i$ but not $T_i$ or vice 
versa) in the solution is not helpful either, hence implying that any $\alpha$-approximation algorithm for set cover needs to find the sets $S_{\istar}$ and $T_{\istar}$. 

The next step is to prove the lower bound for the above distribution. Unlike the lower bound in the one-way model that was based on hiding
 the content of the set $S_{\istar}$, here we need to argue that in fact the index $\istar$ itself is hidden from the players (as otherwise, one more round of
 communication can reveal the content of the sets $S_{\istar}$ and $T_{\istar}$ as well). Similar to~\cite{AssadiKL16}, we also use the information complexity paradigm to prove the communication lower bound
 for this distribution. We embed different instances of the well-known \emph{set disjointness} problem in each pair $(S_i,T_i)$ such that all embedded instances are \emph{intersecting} except for the instance for 
 $S_{\istar}$ and $T_{\istar}$ which is \emph{disjoint}. As we seek a direct-sum style argument for two-way protocols, we need a more careful argument than the one in~\cite{AssadiKL16} that was tailored for one-way protocols. In particular, we now use the notion of \emph{internal} information complexity (as opposed to \emph{external} information complexity used in~\cite{AssadiKL16}) that allows us to use the powerful techniques 
 developed in~\cite{BarakBCR10,Braverman12,BravermanR11} to obtain the direct-sum result. 
 
 Finally, we need to lower bound the information complexity of the set disjointness problem on the specific distribution induced by the set cover instances. The set cover distribution is designed in a way to ensure 
 that the distribution of underlying set disjointness instances matches the known hard input distributions for this problem. However, there is a subtlety here; known information complexity lower bounds for set disjointness (that 
 we are aware of) are all over distributions that are supported only on
 disjoint sets, i.e., $\Yes$-instances of the problem (see, e.g.,~\cite{Bar-YossefJKS02-S,BravermanGPW13,WeinsteinW15})\footnote{We remark that this is not just a coincidence and in fact is crucial for performing the typical 
 reduction to the AND problem used in proving the lower bound for set disjointness, see, e.g.,~\cite{WeinsteinW15} for more details.}. However, for our purpose, we need to lower bound the information cost of set disjointness protocols 
 on distributions that are intersecting. We achieve this using an application of the ``information odometer'' of~\cite{BravermanW15} (and subsequent work in~\cite{GoosJP015}) 
 to relate the information cost of the protocols on \Yes and \No instances of the problem together and obtain the result.

We are not done though, as we seek a lower bound for random arrival streams and for this, we extend the previous communication complexity lower bound to the 
case when the input sets are partitioned randomly across the players, in a similar way as done in previous work~\cite{AssadiKL16} (itself based on~\cite{ChakrabartiCM08}). 
There are however some technical differences needed to execute this approach in our two-way communication model
in compare to the one-way model in~\cite{AssadiKL16} (see Lemma~\ref{lem:dist-random} for details).

\newcommand{\REM}[1]{}

\renewcommand{\prot}{\ensuremath{\pi}}
\renewcommand{\Prot}{\ensuremath{\Pi}}

\section{Preliminaries}\label{SEC:PRELIM}

\paragraph{Notation.} For any integer $a \geq 1$, we let $[a]:=\set{1,\ldots,a}$. We say that a set $S \subseteq [n]$ with $\card{S}=s$ is a \emph{$s$-subset} of $[n]$. 
For a $k$-dimensional tuple $X = (X_1,\ldots,X_k)$ and index $i \in [k]$, we define $X^{<i}:= (X_1,\ldots,X_{i-1})$ and $X^{-i}:=(X_1,\ldots,X_{i-1},X_{i+1},\ldots,X_k)$.

We use capital letters to denote random variables. For a random variable $A$, $\supp{A}$ denotes the support of $A$ and $\card{A} := \log{\card{\supp{A}}}$. 
We use $``A \perp B \mid C"$ to mean that the random variables $A$ and $B$ are independent conditioned on $C$. The notation $``A \in_R U"$ indicates that $A$ is chosen uniformly at random from the set $U$.   

We denote the \emph{Shannon Entropy} of a random variable $A$ by $\HH(A)$ and the \emph{mutual information} of two random variables $ A$ and $ B$ by
$\II( A : B) = \HH( A) - \HH( A \mid  B) = \HH( B) - \HH( B \mid  A)$. If the distribution
$\dist$ of the random variables is not clear from the context, we use $\HH_\dist( A)$
(resp. $\II_{\dist}( A : B)$). Appendix~\ref{app:info} summarizes the relevant information theory tools that we use in this paper. 

\paragraph{Concentration bounds.} We use the following standard version of Chernoff bound (see, e.g.,~\cite{ConcentrationBook}). 

\begin{proposition}\label{prop:chernoff}
	Let $X_1,\ldots,X_n$ be $n$ independent random variables taking values in $[0,1]$ and let $X:= \sum_{i=1}^{n} X_i$. Then, for any $0 \leq \eps \leq 1$,
	\[ \Pr{\card{X - \Ex{X}} > \eps \cdot \Ex{X}} \leq 2 \cdot \exp\paren{-\frac{\eps^{2}\cdot\Ex{X}}{2}} \]
\end{proposition}

We also prove the following useful auxiliary lemma that upper bounds the number of elements that a collection of large random sets can cover. 

\begin{lemma}\label{lem:coverage-lemma}
	Let $\SS = \set{S_1,\ldots,S_k}$ be a collection of $(n-s)$-subsets of $[n]$ that are chosen independently and uniformly at random. Suppose $U \subseteq [n]$ is another set chosen independent of 
	$\SS$; if $k = o(e^{s})$, then, 
	\begin{align*}
		\Pr{\card{U \setminus \paren{S_1 \cup \ldots \cup S_k}} < \frac{\card{U}}{2} \cdot \paren{\frac{s}{2n}}^{k}} < 2 \cdot \exp \paren{ - \frac{\card{U}}{8} \cdot \paren{\frac{s}{2n}}^{k}}
	\end{align*}
\end{lemma}

We first briefly explain the bounds in Lemma~\ref{lem:coverage-lemma}. Note that each element $e \in [n]$, is not covered by a set $S_i \in \SS$ w.p. $\frac{s}{n}$ (as $S_i$ is a random set of size $(n-s)$). Moreover, 
since the sets are chosen independent of each other, the probability that $e$ is not covered by $\SS$ is $\paren{\frac{s}{n}}^{k}$. Hence, in expectation $\card{U} \cdot \paren{\frac{s}{n}}^{k}$ elements in $U$ are not 
covered by $\SS$. We then wish to argue, by means of some concentration bound, that with a very high probability the number of elements not covered by $\SS$ is at least half
of this number (notice that the bounds in the lemma statement are quite similar but not exactly equal to this quantity). 

However, there is an important subtlety here. The random variables defined in the above process 
are \emph{negatively correlated} and hence one cannot readily use a Chernoff-Hoeffding bound (or even similar variants defined for negatively 
correlated random variables) to bound this probability. This is because we need to bound the probability of the sum of these random variable being too small as opposed to being too large which
already follows from known results (see, e.g.,~\cite{PanconesiS97,ImpagliazzoK10})\footnote{Note that in general, Chernoff bound type inequalities do not hold for bounding the sum of negatively random variables from below.}. 
In the following, we show how to get around this using a careful coupling argument.  

\begin{proof}[Proof of Lemma~\ref{lem:coverage-lemma}]
	
	For any element $e \in U$, define the random variable $X_e \in \set{0,1}$ which is $1$ iff $e \notin S_1 \cup \ldots \cup S_k$. Define $X:= \sum_{e \in U} X_e$; notice that $X$ denotes the number of elements in $U$
	that are not covered by $\SS$. Our goal is then to lower bound the value of $X$. Note that the random variables $X_e$ are \emph{negatively correlated} and hence, as stated earlier, we \emph{cannot} use Chernoff bound (or 
	its generalizations to negatively correlated random variables) to \emph{lower bound} the value of $X$. 
	
	To get around this, we slightly change the distribution each set is chosen from, prove the result in that case, and then relate that distribution to the original distribution of the sets in $\SS$. Formally, let $\dist$ be the distribution of 
	from which the sets in $\SS$ are chosen. Consider the following distribution $\dist'$: we create each set $S_i$ (for $i \in [k]$) by removing each element in $[n]$ from $S_i$ independently and 
	uniformly at random w.p. $p = \frac{s}{2n}$. 
	
	We lower bound the value of the random variable $X$ under this new distribution. We first have, 
	\begin{align*}
		\EX_{\dist'}\bracket{X} &= \sum_{e \in U} \PR_{\dist'}\paren{X_e = 1} = \card{U} \cdot p^{k}  = \card{U} \cdot \paren{\frac{s}{2n}}^{k}
	\end{align*}
	For simplicity, define $\eta:= \card{U} \cdot \paren{\frac{s}{2n}}^{k}$. An important property of $\dist'$ is that now all random variables $X_e$ are \emph{independent} of each other. Hence, we can apply Chernoff bound as 
	follows,
	\begin{align}
		\PR_{\dist'}\paren{X < \eta/2} = \PR_{\dist'}\paren{X < \Ex{X}/2} \leq e^{-{\eta}/{8}}
	\end{align}
	We now argue that $\PR_{\dist}\paren{X < \eta/2}$ is in fact very close to $\PR_{\dist'}\paren{X < \eta/2}$. 
	
	Fix a set $S_i \in \SS$. For each $e \in [n]$, define a random variable $Y_e \in \set{0,1}$ which is $1$ iff $e \notin S_i$. Let $Y = \sum_{e \in [n]} Y_e$, i.e., the number of elements missing from $S_i$. 
	Note that $\EX_{\dist'}\bracket{Y} = \frac{s}{2}$. Under the distribution $\dist'$, for each set $S_i \in \SS$, each
	element $e \in [n]$ belongs to $S_i$ independently; hence a simple application of Chernoff bound ensures that:
	\begin{align}
		\PR_{\dist'}\paren{\card{S_i} < (n-s)} = \PR_{\dist'}\paren{Y > 2 \cdot \Ex{Y}} < e^{-s/2} \label{eq:event-not-large}
	\end{align}
	Define $\event$ as the event that all sets $S_i$ has size at least $(n-s)$; by Eq~(\ref{eq:event-not-large}) and a union bound, $\PR_{\dist'}\paren{\event} \geq 1-k\cdot e^{-s/2} \geq \frac{1}{2}$ (as $k = o(e^{s})$). 
	Notice that to sample a set system from $\dist$, we can first sample a set system from $\dist' \mid \event$ and then make the size of each set exactly equal to $(n-s)$ by removing the extra elements uniformly at random;
	this process does not increase the coverage of the original set system sampled from $\dist'$ (or equivalently decrease the value of $X$). Hence, 
	\begin{align*}
		\PR_{\dist}\paren{X < \eta/2} \leq \PR_{\dist'}\paren{X < \eta/2 \mid \event} \leq \frac{\PR_{\dist'}\paren{X < \eta/2}}{\PR_{\dist'}\paren{\event}} \leq 2 \cdot e^{-\eta/8}
	\end{align*}
	By substituting the value of $\eta$, we obtain the desired bound. 
\end{proof}

\subsection{Communication Complexity and Information Complexity}\label{sec:cc-ic}

Communication complexity and information complexity play an important role in our lower bound proofs. 
We now provide necessary definitions for completeness.

\paragraph{Communication complexity.} Our lowers bounds for streaming algorithms
are established via communication complexity lower bounds. 
We use standard definitions of the \emph{two-party communication} model introduced by Yao~\cite{Yao79}; see~\cite{KN97} for 
an extensive overview of communication complexity. 

Let $P$ be a relation with domain $\mathcal{X} \times \mathcal{Y} \times \mathcal{Z}$.  Alice receives an input $X
\in \mathcal{X}$ and Bob receives $Y \in \mathcal{Y}$, where $(X,Y)$ are chosen from a
joint distribution $\dist$ over $\mathcal{X} \times \mathcal{Y}$. They communicate with each other by exchanging messages such that each message
 depends only on the private input of the player sending the message and the already communicated messages. The last message communicated is the answer
 $Z$ such that $(X,Y,Z) \in P$. We allow players to have access to both public and private randomness.  

We use $\prot$ to denote a protocol used by the players. We always assume that the protocol $\prot$ can be randomized (using both public and
private randomness), \emph{even against a prior distribution $\dist$ of inputs}. For any
$0 < \delta < 1$, we say $\prot$ is a $\delta$-error protocol for $P$ over a distribution
$\dist$, if the probability that for an input $(X,Y)$, $\prot$ outputs some $Z$ where $(X,Y,Z) \notin P$ is at most
$\delta$ (the probability is taken over the randomness of both the distribution and the protocol).

\begin{definition}
  The \emph{communication cost} of a protocol $\prot$ for a problem $P$ on an input
  distribution $\dist$, denoted by $\norm{\prot}$, is the worst-case  bit-length of the transcript
  communicated between Alice and Bob in the protocol $\prot$, when the inputs are chosen from $\dist$.
   \newline The \emph{communication complexity} $\CC{P}{\dist}{\delta}$ of a
  problem $P$ with respect to a distribution $\dist$ is the minimum communication cost of
  a $\delta$-error protocol $\prot$ over $\dist$.
\end{definition}

\paragraph{Information complexity.} There are several possible definitions of information
complexity of a communication problem that have been considered depending on the application (see, e.g.,~\cite{Bar-YossefJKS02,BarakBCR10,BravermanR11,ChakrabartiSWY01,Bar-YossefJKS02-S}).  
We use the notion of \emph{internal information complexity}~\cite{BarakBCR10} that measures the average amount of (Shannon) information each player learns about the input of the other player
by observing the transcript of the protocol. Formally, 
\begin{definition}
  Consider an input distribution $\dist$ and a protocol $\prot$ (for some problem
  $P$). Let $(X,Y) \sim \dist$ be the input of Alice and Bob and assume $\Prot:= \Prot(X,Y)$ denotes 
  the transcript of the protocol \emph{concatenated} with the public randomness $R$ used by $\prot$. 
  The \emph{(internal) information cost} $\ICost{\prot}{\dist}$ of a protocol $\prot$ with respect to
  $\dist$ is then $\II_\dist(\Prot : X \mid Y) + \II_{\dist}(\Prot : Y \mid X)$. 
  \newline The \emph{information complexity} $\IC{P}{\dist}{\delta}$ of $P$ with respect to a distribution $\dist$ is
  the minimum $\ICost{\prot}{\dist}$ taken over all $\delta$-error protocols
  $\prot$ for $P$ over $\dist$.
\end{definition}
Note that any public coin protocol is a distribution over private coins protocols, obtained by
first using public randomness to sample a random string $R=r$ and then running the
corresponding private coin protocol $\prot^r$. We also use $\Prot^r$ to denote the transcript of the protocol $\prot^r$. 
We have the following well-known claim.
\begin{claim}\label{clm:public-random}
	 For any distribution $\dist$ and any protocol $\prot$, let $R$ be the public randomness used in $\prot$; then, $\ICost{\prot}{\dist} = \II_{\dist}(\Prot : X \mid Y, R) + \II_{\dist}(\Prot : Y \mid X, R)$. 
\end{claim}
\begin{proof} 
\begin{align*}
	\ICost{\prot}{\dist} &= \II(\Prot : X \mid Y) + \II(\Prot : Y \mid X) \\
	&= \II(\Prot,R : X \mid Y) + \II(\Prot,R : Y \mid X) \tag{$\Prot$ denotes the transcript and the public randomness} \\
	&= \II(R : X \mid Y)  + \II(\Prot : X \mid Y , R) + \II(R : Y \mid X) + \II(\Prot : Y \mid X,R) \tag{chain rule of mutual information, \itfacts{chain-rule}} \\
	&= \II(\Prot : X \mid Y , R) + \II(\Prot : Y \mid X,R) 
\end{align*}
The last equality is because {$\II( R : X \mid Y) = \II(R : Y \mid X) = 0$ since $R \perp X,Y$ and \itfacts{info-zero}}. 
\end{proof}

The following well-known proposition relates communication complexity and internal information complexity (see, e.g.,~\cite{BravermanR11} for a proof). 
\begin{proposition}\label{prop:cc-ic}
  For any distribution $\dist$ and any protocol $\prot$: $\ICost{\prot}{\dist} \leq \norm{\prot}$. Moreover, for any parameter $0 < \delta < 1$: $ \IC{P}{\dist}{\delta} \leq \CC{P}{\dist}{\delta}$.
\end{proposition}

\subsection{The Set Disjointness Problem}\label{sec:disjointness} 

We shall use the well-known \emph{set-disjointness} communication problem (denoted by \Disj) in proving Result~\ref{res:sc-lower}. 
Fix an integer $t \geq 1$; in $\Disj_t$, Alice and Bob are given two sets $A \subseteq [t]$ and $B \subseteq [t]$, and
their goal is to return \Yes if $A \cap B = \emptyset$ and \No otherwise. 

The following is a known hard distribution for $\Disj_t$. 

\textbox{Distribution \distDisj. {\textnormal{A hard input distribution for $\Disjt$.}}} {
\begin{itemize}
	\item Start with $A = B = [t]$. 
	\item For each element $e \in [t]$ independently: w.p. $1/3$ drop $e$ from both $A$ and $B$, w.p. $1/3$ drop $e$ from $A$, and w.p. $1/3$ drop $e$ from $B$. 
	\item Pick $Z \in_R \set{0,1}$ uniformly at random. If $Z = 1$, pick a uniformly at random element $\estar \in [t]$ and let $A$ and $B$ both contain $\estar$ (if $Z=0$, keep the sets as before). 
\end{itemize}
}

We further use $\distDisjY$ and $\distDisjN$ to denote, respectively, the
distribution of \Yes and \No instances of \Disj on \distDisj; in other words, $\distDisjY:= \paren{\distDisj \mid Z=0}$ and $\distDisjN := \paren{\distDisj \mid Z=1}$. 

The following proposition on the information complexity of $\Disj$ is well-known (see, e.g.,~\cite{Bar-YossefJKS02-S,BravermanGPW13}). 

\begin{proposition}\label{prop:disjointness}
	For any $\delta < 1/2$ and any $\delta$-error protocol $\protDisj$ of $\Disj_t$ on the distribution $\distDisj$, 
	\[\ICost{\protDisj}{\distDisjY} = \Omega(t). \]
\end{proposition}

\section{The Space-Approximation Tradeoff for Set Cover} \label{SEC:SET-COVER}

We prove our main result on the space-approximation tradeoff for the streaming set cover problem in this section. Formally,

\begin{theorem}\label{thm:sc-lower}
  For any $\alpha = o(\log{n}/\log\log{n})$, $m = \poly(n)$, and $p \geq 1$, any \emph{randomized} algorithm that can make $p$ passes over any collection of $m$ subsets of $[n]$ presented in a 
  {random order stream} and outputs an $\alpha$-approximation to the optimal value of the set cover problem w.p. larger than $3/4$ (over the randomness of both the
  stream order and the algorithm) must use $\Omgt(mn^{\frac{1}{\alpha}}/p)$ space. 
\end{theorem}

Theorem~\ref{thm:sc-lower} formalizes Result~\ref{res:sc-lower} in the introduction. We further prove that the tradeoff achieved in Theorem~\ref{thm:sc-lower} is in fact tight up to logarithmic factors; this is achieved by performing some proper modifications to the algorithm of~\cite{HarPeledIMV16}. 
Formally, 

\begin{theorem}\label{thm:sc-upper}
	There exists a streaming algorithm that for any integer $\alpha \geq 1$, and any parameter $\eps > 0$, with high probability, computes an $(\alpha + \eps)$-approximation to the streaming set cover problem
	using $(2\alpha+1)$ passes over the stream in adversarial order and $\Ot(mn^{1/\alpha}/\eps^{2} + n/\eps)$ space.  
\end{theorem}

We emphasize that the main contribution of the paper is in proving Theorem~\ref{thm:sc-lower}; we mainly present Theorem~\ref{thm:sc-upper} to prove a matching upper bound on the bounds in Theorem~\ref{thm:sc-lower}, 
hence establishing a tight space-approximation tradeoff for the streaming set cover problem. 

The rest of this section is mainly devoted to the proof of Theorem~\ref{thm:sc-lower}. 
We start by introducing some notation. In Section~\ref{sec:hard-dist}, we introduce a hard input distribution
for the set cover problem in adversarial streams. We prove a lower bound for this distribution in Section~\ref{sec:sc-lower}. We extend this lower bound to random arrival streams in 
Section~\ref{sec:sc-random-arrival} and finish the proof of Theorem~\ref{thm:sc-lower}. Section~\ref{sec:sc-upper} contains the proof of Theorem~\ref{thm:sc-upper}. 

\paragraph{Notation.} To prove Theorem~\ref{thm:sc-lower}, we prove a lower bound on the communication complexity of the set cover problem: 
Fix a (sufficiently large) value for $n$, $m = \poly(n)$, and $\alpha = o(\log{n}/\log\log{n})$; in this section, \SetCover refers to the problem of $\alpha$-approximating the optimal value of 
the set cover problem with $2m$ sets\footnote{\label{footnote:2m}To simplify the exposition, we use $2m$ instead of $m$ as the number of sets.} defined over the universe $[n]$ in the two-player communication model, whereby
the sets are partitioned between Alice and Bob.

\subsection{A Hard Input Distribution for \SetCover}\label{sec:hard-dist}

Let $t$ be an integer to be determined later; we use the distribution $\distDisj$ for $\Disj_t$ (introduced in Section~\ref{sec:disjointness}) to design 
a hard input distribution for $\SetCover$. Before that, we need a simple definition. 

\begin{definition}[Mapping-extension]\label{def:p-extension}
	For the two sets $[t]$ and $[n]$, we define a \emph{mapping-extension} of $[t]$ to $[n]$ as a
	function $f: [t] \mapsto 2^{[n]}$, whereby for each $i \in [t]$, $f(i) \subseteq [n]$ is mapped to $n/t$ \emph{unique} elements in $[n]$. 
	Similarly, for any set $A \subseteq [t]$, we abuse the notation and define $f(A):= \bigcup_{i \in A}f(i)$. 
\end{definition}

We are now ready to define our hard input distribution for \SetCover.

\textbox{Distribution \distSC. {\textnormal{A hard input distribution for $\SetCover$.}}} {

\medskip

\textbf{Notation.} Let $t:= 2^{-15} \cdot \paren{\frac{n}{\log{m}}}^{\frac{1}{\alpha}}$ and $\FC$ be the set of all mapping-extensions of $[t]$ to $[n]$.
\begin{itemize}
	\item For each $i \in [m]$: 
	\begin{itemize}
		\item Let $(A_i,B_i) \sim \distDisjN$ for $\Disjt$ and pick $f_i \in_{R} \FC$ uniformly at random.
		\item Let $S_i = [n] \setminus f_i(A_i)$ and $T_i := [n] \setminus f_i(B_i)$. 
	\end{itemize}
	\item Pick $\theta \in_R \set{0,1}$ uniformly at random. If $\theta = 0$, do nothing, otherwise: 
	\begin{itemize}
		\item Sample $\istar \in_R [m]$ uniformly at random.
		\item Resample $(A_{\istar},B_{\istar}) \sim \distDisjY$ for $\Disjt$ and redefine $S_{\istar}$ and $T_{\istar}$ as before using the new pair $(A_{\istar},B_{\istar})$. 
	\end{itemize}
	\item Let the input to Alice and Bob be $\SS:=\set{S_i}_{i\in[m]}$ and $\ST:=\set{T_i}_{i\in[m]}$, respectively. 
\end{itemize}
}

In the following, we use $Z$ to denote any set in $\SS \cup \ST$, i.e., when it is not relevant whether it belongs to $\SS$ or $\ST$. For a collection of 
sets $\SZ = \set{Z_1,\ldots,Z_\ell}$, we use $C(\SZ)$ to denote the set of elements that $\SZ$ covers, i.e., $C(\SZ):= \bigcup_{i=1}^{\ell} Z_i$.  We say that $\SZ$ is a 
\emph{singleton-collection}, if for any $i \in [m]$, at least one of $S_i$ or $T_i$ is \emph{not} present in $\SZ$. In contrast, we say
that $\SZ$ is a \emph{pair-collection}, if for all $i \in [m]$, $S_i \in \SZ$ iff $T_i \in \SZ$ as well.

\begin{remark}\label{rem:dist-sc}
A few remarks are in order: 
\begin{enumerate}[(i)]
	\item \label{part:set-size} W.h.p., for any $i \in [m]$, $\card{S_i} = 2n/3 \pm o(n)$ and $\card{T_i} = 2n/3 \pm o(n)$.\newline
	(Proof. follows from the definition of the distribution $\distDisj$ and Chernoff bound). 

	\item \label{part:random-sets} For any $i \in [m]$, conditioned on $\card{S_i} = \ell$, the set $S_i$ is chosen uniformly at random from all $\ell$-subsets of $[n]$; similarly for $T_i$
	
	\item \label{part:missing-elements} For any $i \in [m]$, $S_i \cup T_i = [n] \setminus f_i(A_i \cap B_i)$. Moreover, whenever $(A_i,B_i) \sim \distDisjN$,
	the set $f_i(A_i \cap B_i)$ is a $(n/t)$-subset of $[n]$ chosen uniformly at random. \newline
	(Proof. the first part follows from the fact that $f_i$ maps each $j \in [t]$ to unique elements; the second part is by the random choice of $f_i \in_R \FC$ and the fact that $\card{A_i \cap B_i}=1$ in this case). 
	
	\item \label{part:independent-sets} Whenever $\theta = 0$, for any $i \neq j$, the sets $Z_i \in \set{S_i,T_i}$ and $Z_j \in \set{S_j,T_j}$ are chosen independent of each other ($Z_i \perp Z_j$).  
	
\end{enumerate}
\end{remark}

Let $\opt(\SS,\ST)$ denote the size of an optimal set cover in the instance $(\SS,\ST)$. It follows from Remark~\ref{rem:dist-sc}-(\ref{part:missing-elements}) that 
whenever $\theta=1$ in the distribution $\distSC$, $\opt(\SS,\ST) = 2$; simply take $S_{\istar}$ and $T_{\istar}$ and since $A_{\istar} \cap B_{\istar} = \emptyset$, they cover the whole universe.  
In the following, we prove that when $\theta=0$, $\opt(\SS,\ST)$ is relatively large. This implies that any $\alpha$-approximation protocol for \SetCover has to essentially determine the value of $\theta$. In
the next section, we prove that this task requires a large communication by the players. 

\begin{lemma}\label{lem:dist-sc-opt}
	For $(\SS,\ST) \sim \distSC$: 
	\[\Pr{\opt(\SS,\ST) > 2\alpha \mid \theta = 0} = 1-o(1).\] 	
\end{lemma}
\begin{proof}
	
	Let $\SC$ be any collection of $2\alpha$ sets from $(\SS,\ST)$. We bound the probability that $\SC$ covers the universe $[n]$ entirely, i.e., is a feasible set cover, and then use a union bound
	on all possible choices for $\SC$ to finalize the proof. In the following, we condition on the event $\event_1$ that states that $\card{S_i} \leq 3n/4$ and $\card{T_i} \leq 3n/4$ for all $i \in [m]$
	(which happens with probability $1-o(1)$ by Remark~\ref{rem:dist-sc}-(\ref{part:set-size})). 
	
	Partition the collection $\SC$ into a pair-collection $\SC_P$, and a singleton-collection $\SC_{S}$ (this partitioning is always possible and unique by definition). 
	We first lower bound the number of elements that are not covered by the singleton-collection: 
	
	\begin{claim}\label{clm:singletons}
		$\Pr{\card{\bar{C(\SC_S)}} \leq \frac{n}{2^{6\alpha+1}} \mid \event_1} \leq 1-\frac{1}{m^{\omega(\alpha)}}$. 
	\end{claim}
	\begin{proof}
		Let $\SC_{S} := \set{Z_1,\ldots,Z_k}$; clearly $k = \card{\SC_{S}} \leq \card{\SC} = 2\alpha$. Without loss of generality, we assume that $k = 2\alpha$. 
		By conditioning on the event $\event_1$ and Remark~\ref{rem:dist-sc}-(\ref{part:random-sets}), we know that each $Z_i$ is an $\ell_i$-subset of $[n]$, for some $\ell_i \leq 3n/4$, chosen uniformly at random
		from all $\ell_i$-subsets of $[n]$. Again without loss of generality, we simply increase the size of each $Z_i$ so that they all have size exactly $3n/4$. 
		Moreover, since no two sets $S_i$ and $T_i$ are both simultaneously present in $\SC_{S}$, by Remark~\ref{rem:dist-sc}-(\ref{part:independent-sets}), all
		sets in $\SC_{S}$ are chosen independent of each other. 
		
		Consequently, by Lemma~\ref{lem:coverage-lemma}, for $U = [n]$, $s = n/4$, and collection $\SC_{S}$, we have, 
		\begin{align*}
				\Pr{\card{\bar{C(\SC_S)}} < \frac{{n}}{2} \cdot \paren{\frac{1}{8}}^{2\alpha}  \mid \event_1} < 2 \cdot \exp\paren{ - \frac{{n}}{8} \cdot \paren{\frac{1}{8}}^{2\alpha}} 
		\end{align*}
		A simplification of the above equation, plus using the fact that $\alpha = o(\log{n}/\log\log{n})$, and hence $n/2^{\Theta(\alpha)} = \omega(\alpha\log{m})$, proves the final result. 
	\end{proof}
	
	Let $\event_2$ be the event that $\card{\bar{C(\SC_S)}} \geq \frac{n}{2^{6\alpha+1}}$; in the following, we condition on this event. Now consider the sets in the pair-collection $\SC_{P}$. 
	For any pair $(S_i,T_i) \in \SC_P$, we define $C_i := S_i \cup T_i$. Note that there are at most $\alpha$ different possible sets $C_i$. 
	By Remark~\ref{rem:dist-sc}-(\ref{part:missing-elements}), the sets $C_i$'s are random sets of size $(n-n/t)$, and by Remark~\ref{rem:dist-sc}-(\ref{part:independent-sets}), they are chosen independent of 
	each other. By Lemma~\ref{lem:coverage-lemma}, for $U = \bar{C(\SC_S)}$, $s = n/t$, and collection of sets $C_i$'s, we have,
	\begin{align*}
		&\Pr{{U \setminus \paren{C(\SC_P)}} = \emptyset \mid \event_1,\event_2} \leq  2 \cdot \exp \paren{- \frac{n}{2^{6\alpha+4}} \cdot \paren{\frac{1}{2t}}^{\alpha}} \leq \frac{1}{m^{3\alpha}}
	\end{align*}
	
	We can now conclude, 
	\begin{align*}
		\Pr{\opt(\SS,\ST) \leq 2\alpha} &\leq \Pr{\bar{\event_1}} + \Pr{\text{$\exists~\SC$ that covers $[n]$} \mid \event_1} \\
		&\leq \Pr{\bar{\event_1}} + \sum_{\SC} \paren{\Pr{\bar{\event_2} \mid \event_1} + \Pr{C(\SC) = [n] \mid \event_1,\event_2}} \\
		&\leq o(1) + {{m} \choose {2\alpha}} \cdot \paren{\frac{1}{m^{\omega(\alpha)}} + \frac{1}{m^{3\alpha}}} = o(1) 
	\end{align*}
	proving the lemma. 
\end{proof}

\newcommand{\ProtDisj}{\Prot_{\Disj}}
\newcommand{\ProtSC}{\Prot_{\textnormal{\textsf{SC}}}}

\subsection{The Lower Bound for the Distribution $\distSC$}\label{sec:sc-lower}

Throughout this section, fix $\protSC$ as a $\delta$-error protocol for \SetCover on the distribution $\distSC$. We first show that protocol 
$\protSC$ is essentially solving $m$ copies of the $\Disj_t$ problem on the distribution $\distDisj$ (for the parameter $t$ in the distribution $\distSC$) and then use a direct-sum style argument (similar in  
spirit to the ones in~\cite{BarakBCR10,Braverman12,BravermanR11}) to argue that the information cost of $\protSC$ shall  be $m$ times larger than the information
complexity of solving $\Disj_t$. However, to make the direct-sum argument work, we can only consider $\protSC$ on the distribution $\distSC \mid \theta = 0$, i.e., when \emph{all} underlying $\Disj_t$ instances
are sampled from $\distDisjN$. Consequently, we can only lower bound the information cost of $\protSC$ based on the information complexity of $\Disj_t$ on the distribution $\distDisjN$.

\begin{lemma}\label{lem:direct-sum}
	There exists a $\paren{\delta+o(1)}$-protocol $\protDisj$ for $\Disjt$ on the distribution $\distDisj$ such that:
	\begin{enumerate}
	\item \label{part:ds-ic} $\ICost{\protDisj}{\distDisjN} = \frac{O(1)}{m} \cdot \ICost{\protSC}{\distSC}$.
	\item \label{part:ds-cc} $\norm{\protDisj} = \norm{\protSC}$. 
	\end{enumerate}
\end{lemma}
\begin{proof}
	We design the protocol $\protDisj$ as follows: 
	\textbox{Protocol $\protDisj$. \textnormal{The protocol for solving $\Disjt$ using a protocol $\protSC$ for $\SetCover$.}}{
	\medskip
	\\ 
	\textbf{Input:} An instance $(A,B) \sim \distDisj$.  \textbf{Output:} \Yes if $A \cap B = \emptyset$ and \No otherwise. \\
	\algline
	
	\begin{enumerate}
		\item Using public randomness, the players sample an index $\istar \in_R [m]$ and $m$ mapping-extensions $f_1,\ldots,f_m$ independently and uniformly at random from $\FC$. 
		\item Using public randomness, the players sample the sets $A^{<\istar}$ and $B^{>\istar}$ each from $\distDisjN$ independently. 
		\item Using private randomness, Alice samples the sets $A^{>\istar}$ such that $(A_j,B_j) \sim \distDisjN$ (for all $j > \istar$); similarly Bob samples the sets $B^{<\istar}$. 
		\item The players construct the collections $\SS:= \set{S_1,\ldots,S_m}$ and $\ST:= \set{T_1,\ldots,T_m}$ by setting $S_i := [n] \setminus f_i(A_i)$ and $T_i:= [n] \setminus f_i(B_i)$ (exactly as in
		distribution $\distSC$).
		\item The players solve the \SetCover instance using \protSC and output \No iff $\protSC$ estimates $\opt(\SS,\ST) \leq 2\alpha$ and \Yes otherwise. 
	\end{enumerate}
	}

	It is easy to see that the distribution of instances $(\SS,\ST)$ created in the protocol $\protDisj$ matches the distribution $\distSC$ for \SetCover exactly. Moreover, by Lemma~\ref{lem:dist-sc-opt}, 
	$\opt(\SS,\ST) > 2\alpha$ w.p. $1-o(1)$, whenever $(A,B) \sim \distDisjN$ and $\opt(\SS,\ST) = 2$ whenever $(A,B) \sim \distDisjY$. Consequently, since $\protSC$ is an $\alpha$-approximation protocol,
	\begin{align*}
		\PR_{\distDisj}\paren{\protDisj~\errs} \leq \PR_{\distSC}\paren{\protSC~\errs} + o(1) \leq \delta + o(1)
	\end{align*}
	and hence $\protDisj$ is indeed a $\paren{\delta+o(1)}$-error protocol for $\Disj$ on the distribution $\distDisj$. Moreover, it is clear that the communication cost of $\protDisj$ is at
	most the communication cost of \protSC. We now prove the bound on the information cost of this protocol. 
	
	Our goal is to bound the information cost of $\protDisj$ whenever the instance $(A,B)$ is sampled from $\distDisjN$. Let $F$ be a random variable denoting the tuple $(f_1,\ldots,f_m)$,
	$I$ be a random variable for $\istar$ and $R$ be the set of public randomness used by the players. By Claim~\ref{clm:public-random}, 
	\begin{align*}
		\ICost{\protDisj}{\distDisjN} &= \II_{\distDisjN}(\ProtDisj : A \mid B,R) + \II_{\distDisjN}(\ProtDisj: B \mid A,R)
	\end{align*}
	We now bound the first term in the RHS above (the second term can be bounded exactly the same). 
	\begin{align*}
		\II_{\distDisjN}(\ProtDisj : A \mid B,R) &= {\II_{\distDisjN}(\ProtDisj : A \mid B,R,I)} \tag{$I$ is chosen using public randomness} \\
		&= \sum_{i=1}^{m} \Pr{I=i} \cdot \II_{\distDisjN}\paren{\ProtDisj : A_i \mid B_i,A^{<i},B^{>i},F,I=i} \tag{$R = (A^{<i},B^{>i},F,I)$} \\ 
		&= \sum_{i=1}^{m} \frac{1}{m} \cdot \II_{\distDisjN}\paren{\ProtDisj : A_i \mid B_i,A^{<i},B^{>i},F} 		
	\end{align*}
	where the last equality is true since conditioned on $(A,B) \sim \distDisjN$, all sets $A_j,B_j$ (for $j \in [m]$) are chosen from $\distDisjN$ and hence are independent of the $``I=i"$ event\footnote{We point out that
	this is the exact reason we need to consider information cost of $\protDisj$ on $\distDisjN$ (instead of $\distDisj$) as otherwise $(A_j,B_j)$'s are \emph{not} independent of $I = i$ and hence this equality would not hold.}. 
	Define $\bA:= (A_1,\ldots,A_m)$ and $\bB:= (B_1,\ldots,B_m)$; we can further derive, 
	\begin{align*}
		\II_{\distDisjN}(\ProtDisj : A \mid B,R) &= \sum_{i=1}^{m} \frac{1}{m} \cdot \II_{\distDisjN}\paren{\ProtDisj : A_i \mid B_i,A^{<i},B^{>i},F} \\
		&\leq \frac{1}{m} \cdot \sum_{i=1}^{m} \II_{\distDisjN}\paren{\ProtDisj : A_i \mid A^{<i},\bB,F} \tag{$A_i \perp B^{<i} \mid \bB,F$ and hence we can apply Fact~\ref{fact:info-increase}} \\
		&= \frac{1}{m} \cdot \II_{\distDisjN}\paren{\ProtDisj :  \bA \mid \bB,F} \tag{chain rule of mutual information, \itfacts{chain-rule}} \\
		&= \frac{1}{m} \cdot \II_{\distDisjN}\paren{\ProtDisj :  \SS \mid \ST,F} \\
		&= \frac{1}{m} \cdot \II_{\distSC}\paren{\ProtSC :  \SS \mid \ST,F,\theta=0} 
	\end{align*}
	where the second last equality is because $\bA$ (resp. $\bB$) and $\SS$ (resp. $\ST$) determine each other conditioned on $F$, and 
	last equality is because the distribution of set cover instances and the messages communicated by the players under $\distDisjN$ and under $\distSC \mid \theta = 0$ exactly matches. 
	
	Moreover, 
	\begin{align*}
		\II_{\distDisjN}(\ProtDisj : A \mid B,R) &\leq \frac{1}{m} \cdot \II_{\distSC}\paren{\ProtSC :  \SS \mid \ST,F,\theta=0} \\
		&\leq \frac{2}{m} \cdot \II_{\distSC}\paren{\ProtSC :  \SS \mid \ST,F,\theta} \tag{by definition of mutual information as $\Pr{\theta = 0} = 1/2$} \\
		&\leq \frac{2}{m} \cdot \paren{\II_{\distSC}\paren{\ProtSC :  \SS \mid \ST,F} + H(\theta)} \tag{by Fact~\ref{fact:bar-hopping}}\\
		&= \frac{2}{m} \cdot \II_{\distSC}\paren{\ProtSC : \SS \mid \ST,F} + \frac{2}{m} \tag{$H(\theta) = 1$ by \itfacts{uniform}}  \\
		&\leq \frac{2}{m} \cdot \II_{\distSC}\paren{\ProtSC :  \SS \mid \ST} + \frac{2}{m} \tag{$\ProtSC \perp F \mid \SS,\ST$ and hence we can apply Fact~\ref{fact:info-decrease}}
	\end{align*}
	By performing the same exact calculation for $\II_{\distDisjN}(\ProtDisj : B \mid A,R)$, we obtain that, 
	\begin{align*}
		\ICost{\protDisj}{\distDisjN} &\leq \frac{2}{m} \cdot \paren{ \II_{\distSC}\paren{\ProtSC : \SS \mid \ST}  +  \II_{\distSC}\paren{\ProtSC : \ST \mid \SS} } + \frac{4}{m} \\
		&= \frac{2}{m} \cdot \ICost{\protSC}{\distSC} + \frac{4}{m} = \frac{O(1)}{m} \cdot \ICost{\protSC}{\distSC} 
	\end{align*}
	where in the last inequality we used the fact that information cost of $\protSC$ is at least $1$. This finalizes the proof of the lemma. 
\end{proof}

Recall that in Lemma~\ref{lem:direct-sum}, we bound the information cost of $\protDisj$ on the distribution $\distDisjN$ (as opposed to $\distDisj$); in the following we prove that 
this weaker bound is still sufficient for our purpose. 

\begin{lemma}\label{lem:disjointness}
	For any $\delta < 1/2$, any $\delta$-error protocol $\protDisj$ for $\Disjt$ on $\distDisj$ with $\norm{\protDisj} = 2^{o(t)}$ has 
	\[\ICost{\protDisj}{\distDisjN} = \Omega(t).\] 
\end{lemma}

	By Proposition~\ref{prop:disjointness}, any $\delta$-error protocol for $\Disjt$ (with $\delta < 1/2$) on $\distDisj$ has $\ICost{\protDisj}{\distDisjY} = \Omega(t)$ (notice again that the information
	 cost is measured on the distribution $\distDisjY$). 
	From this, it is also easy to obtain that 
	$\ICost{\protDisj}{\distDisj} = \Omega(t)$. However, to prove Lemma~\ref{lem:disjointness}, we need to lower bound the information cost
	of $\protDisj$ under the distribution $\distDisjN$. 
	
	To achieve this, we can relate the information costs $\ICost{\protDisj}{\distDisjY}$ and $\ICost{\protDisj}{\distDisjN}$ to each other. The goal is to argue that if there is a large discrepancy in the information cost of $\protDisj$ 
	on $\distDisjY$ and $\distDisjN$, then the information cost of the protocol itself can be used to distinguish between these two cases. We can achieve this goal using an elegant 
	construction of an ``information odometer'' by~\cite{BravermanW15}; informally speaking, the odometer allows the players to ``keep track'' of the amount of information revealed
	 in a protocol (i.e., the information cost of the protocol), while incurring a relatively small additional information cost overhead. 
	 
	Intuitively, we can use the odometer to argue that $\ICost{\protDisj}{\distDisjN} = \Theta(\ICost{\protDisj}{\distDisjY})$ as follows: suppose towards a contradiction that 
	$\ICost{\protDisj}{\distDisjN} = \tau$ for some $\tau = o(\ICost{\protDisj}{\distDisjY})$ and consider a new protocol $\protDisj'$ for $\Disj$ on $\distDisj$ which runs $\protDisj$ and the
	 information odometer for $\protDisj$ in parallel. Whenever
	the odometer estimates the information cost of $\protDisj$ to be larger than $c \cdot \tau$ (for some sufficiently large constant $c$), the players terminate the protocol and declare that the answer for $\Disj$ is $\No$
	(as information cost of $\protDisj$ on $\distDisjN$ is typically not much more than $\tau$, while its information cost on $\distDisjY$ is $\omega(\tau)$). If the cost
	is not estimated more than $c \cdot \tau$ by the end of the protocol, the players output the same answer as in $\protDisj$. As the information cost of the information odometer itself is bounded by $O(\tau)$, 
	this results in protocol $\protDisj'$ to have $\ICost{\protDisj'}{\distDisj} = o(t)$, a contradiction. This argument was first made explicit in~\cite{GoosJP015}. 
		
	\begin{lemma}[Lemma 15 in \cite{GoosJP015}]\label{lem:one-sided-two-sided}
		Fix any function $F$, constants $0 < \eps_1 < \eps_2 < 1/2$, input distribution $\dist$, and define $\distN := \dist \mid F^{-1}(\No)$. 
		For every $\eps_1$-error protocol $\prot$ for $F$ on $\dist$, there exists an $\eps_2$-error protocol $\prot'$ for $F$ on $\dist$ such that: 
		\begin{align*}
			\ICost{\prot'}{\dist} = O\paren{\ICost{\prot}{\distN} + \log{\norm{\prot}}}.
		\end{align*}
	\end{lemma}

	We are now ready to prove Lemma~\ref{lem:disjointness}. 

\begin{proof}[Proof of Lemma~\ref{lem:disjointness}]
	Let $\protDisj'$ be any $\delta'$-error protocol for $\Disj$ on $\distDisj$ for $\delta' < 1/2$. 
	We first prove that $\ICost{\protDisj'}{\distDisj} = \Omega(t)$ using the fact that $\ICost{\protDisj'}{\distDisjY} = \Omega(t)$ as follows: 
	\begin{align*}
		\ICost{\protDisj'}{\distDisj} &= \II_{\distDisj}(\ProtDisj' : A \mid B) +  \II_{\distDisj}(\ProtDisj' : B \mid A) \\
		&\geq \II_{\distDisj}(\ProtDisj' : A \mid B,\theta) +  \II_{\distDisj}(\ProtDisj' : B \mid A,\theta) - 2\HH(\theta) \tag{by Fact~\ref{fact:bar-hopping}}\\
		&\geq \frac{1}{2} \cdot \II_{\distDisj}(\ProtDisj' : A \mid B,\theta = 0) +  \frac{1}{2} \cdot \II_{\distDisj}(\ProtDisj' : B \mid A,\theta=0) - 2 \tag{by definition of mutual information and since $\HH(\theta) = 1$} \\
		&= \frac{1}{2} \cdot \paren{\II_{\distDisjY}(\ProtDisj' : A \mid B) +  \II_{\distDisjY}(\ProtDisj' : B \mid A)} - 2 \tag{$\distDisjY = \distDisj \mid \theta = 0$} \\
		&= \frac{1}{2} \cdot \ICost{\protDisj'}{\distDisjY} - 2 = \Omega(t) \tag{by Proposition~\ref{prop:disjointness}}
	\end{align*}
	
	Now suppose towards a contradiction that $\ICost{\protDisj}{\distDisjN}$ is $o(t)$. We can then apply Lemma~\ref{lem:one-sided-two-sided} for the function $F = \Disj$, $\eps_1 = \delta$ and $\eps_2 = \delta' < 1/2$ to obtain 
	a protocol $\protDisj'$ with $\ICost{\protDisj'}{\distDisj} = O\paren{\ICost{\protDisj}{\distN} + \log{\norm{\protDisj}}}$ which is $o(t)$; a contradiction.
\end{proof}

We now conclude, 

\begin{theorem}\label{thm:sc-cc-lower}
	For any constant $\delta < 1/2$, $\alpha = o(\frac{\log{n}}{\log{\log{n}}})$, and $m = \poly{(n)}$, 
	\begin{align*}
		\CC{\SetCover}{\distSC}{\delta} = \Omgt({mn^{\frac{1}{\alpha}}}).
	\end{align*} 
\end{theorem}
\begin{proof}
	Let $t = \Theta\paren{(\frac{n}{\log{m}})^{\frac{1}{\alpha}}}$ and suppose towards a contradiction that there exists a $\delta$-error protocol $\protSC$ for $\SetCover$ on the 
	distribution $\distSC$ with $\norm{\protSC} = o(mt)$; by Proposition~\ref{prop:cc-ic}, $\ICost{\protSC}{\distSC} = o(mt)$ also.
	By Lemma~\ref{lem:direct-sum}, this implies that there exists a $\paren{\delta+o(1)}$-error protocol $\protDisj$ for $\Disj$ on
	the distribution $\distDisj$ such that $\ICost{\protDisj}{\distDisjN} = o(t)$, and $\norm{\protDisj} = o(mt) \leq 2^{o(t)}$ (since $m = \poly(n)$ and $\alpha = o(\frac{\log{n}}{\log{\log{n}}})$). 
	However, this is in contradiction with Lemma~\ref{lem:disjointness}, implying that $\norm{\protSC} = \Omega(mt)$, hence proving the theorem. 
\end{proof}

As a corollary of Theorem~\ref{thm:sc-cc-lower}, we have that the space complexity of any $\alpha$-approximation streaming algorithm for set cover 
that uses $\polylog{(n)}$ passes on \emph{adversarial streams} is $\Omgt({mn^{\frac{1}{\alpha}}})$. In the next section, we extend this result to random arrival streams
and complete the proof of Theorem~\ref{thm:sc-lower}.

\newcommand{\distSCrnd}{\ensuremath{\distSC^{\textnormal{\textsf{rnd}}}}}

\newcommand{\Gstar}{\ensuremath{G^{\star}}}

\subsection{Proof of Theorem~\ref{thm:sc-lower}}\label{sec:sc-random-arrival}

The distribution $\distSC$ used in the previous section is quite ``adversarial'' and as such is not suitable for proving the lower bound for random arrival streams. 
In order to prove the lower bound in Theorem~\ref{thm:sc-lower} for random arrival streams, we need to relax the adversarial
 partitioning of the sets in the distribution $\distSC$ to a randomized partition.

\textbox{Distribution \distSCrnd. {\textnormal{A random partitioning of the distribution $\distSC$}}} {

\begin{itemize}
	\item Sample the collections $(\SS,\ST) \sim \distSC$. 
	\item Assign each set in $\SS \cup \ST$ to Alice w.p. $1/2$ and the remainings to Bob.  
\end{itemize}
}

We show that even this seemingly easier distribution still captures all the ``hardness'' of distribution $\distSC$. Formally, 
\begin{lemma}\label{lem:dist-random}
	For any constant $\delta < 1/4$, $\alpha = o(\frac{\log{n}}{\log\log{n}})$, and $m = \poly(n)$, 
	\begin{align*}
		\CC{\SetCover}{\distSCrnd}{\delta} = \Omgt(mn^{\frac{1}{\alpha}})
	\end{align*}
\end{lemma}
\begin{proof}
	Let $\SS = \set{S_1,\ldots,S_m}$ and $\ST = \set{T_1,\ldots,T_m}$ be the collections of sets sampled from $\distSC$ in the distribution $\distSCrnd$. For a sampled instance in $\distSCrnd$, we say that the index 
	$i \in [m]$ is \emph{good} iff $S_i$ is given to one player and $T_i$ to another. Let $G \subseteq [m]$ be the collection of all good indices. The index $\istar$ is chosen independent of the 
	random partitioning in $\distSCrnd$, and hence the probability that $\istar \in G$ is exactly $\card{G}/m$. Let $\event$ denote the event that $\card{G} \geq m/2 - o(m)$ and $i \in G$. We have, 
	\begin{align*}
		\Pr{\event} &= \Pr{\card{G} \geq m/2-o(m)} \cdot \Pr{\istar \in G \mid \card{G} \geq m/2-o(m)} \\
		&\geq \Pr{\card{G} \geq \paren{1-o(1)} \cdot \Ex{G}} \cdot \frac{1-o(1)}{2} \geq (1-o(1)) \cdot \frac{1}{2}
	\end{align*}
	where the last inequality is by Chernoff bound. 
	Now fix a $\delta$-error protocol $\protSC$ for \SetCover on the distribution $\distSCrnd$. Then, 
	\begin{align}
		\Pr{\protSC~\errs \mid \event} \leq \frac{\Pr{\protSC~\errs}}{\Pr{\event}} \leq 2\delta +o(1)\label{eq:prob-errs}
	\end{align}
	
	This in particular implies that there exists a set $\Gstar \subseteq [n]$ with $\card{\Gstar} \geq m/2 - o(m)$, such that conditioned on the set of good indices being 
	$\Gstar$ and conditioned on $\istar \in \Gstar$, the probability that $\protSC$ errs is at most $2\delta+o(1)$. Note that conditioned on the aforementioned events, the index $\istar$ is 
	chosen from $\Gstar$ uniformly at random. This implies that the distribution of the input given to Alice and Bob limited to the sets in $\Gstar$ matches the distribution $\distSC$ (with the number of the sets being 
	$2 \cdot \card{\Gstar}$ instead of $2m$). We can then use this to embed an instance of $\SetCover$ over the distribution $\distSC$ into the sets $\Gstar$ and obtain a protocol $\protSC'$ for $\distSC$. 
	
	More formally, the protocol $\protSC'$ works as follows: Given an instance $(\SS',\ST')$ sampled from $\distSC$ (with $\card{\SS'} = \card{\ST'} = \card{\Gstar}$), Alice and Bob use public coins to complete their input
	(i.e., increase the number of the sets to $2m$) by sampling from the distribution $\distSCrnd$ conditioned on $\Gstar$ (this is possible without any communication as the sets outside $\Gstar$ are sampled independent of the 
	 sets in $\Gstar$). The players then run the protocol $\protSC$ on this new instance and return the same answer as this protocol. As the distribution of the \SetCover instances sampled in the protocol $\protSC'$ matches
	 the distribution $\distSCrnd$ conditioned on $\Gstar$ and $\istar \in \Gstar$, by Eq~(\ref{eq:prob-errs}), the probability that $\protSC'$ errs is at most $2\delta+o(1)$. Since $\delta < 1/4$, we obtain a $\delta'$-error protocol 
	 for $\SetCover$ on the distribution $\distSC$ with $2\card{\Gstar} = \Theta(m)$ sets and universe of size $n$, for a constant $\delta' < 1/2$. Consequently, by Theorem~\ref{thm:sc-cc-lower}, 
	 $\norm{\protSC} = \norm{\protSC'} = \Omgt(\card{\Gstar} \cdot n^{\frac{1}{\alpha}}) = \Omgt(mn^{\frac{1}{\alpha}})$, proving the lemma.
\end{proof}

We are now ready to prove Theorem~\ref{thm:sc-lower}. 

\begin{proof}[Proof of Theorem~\ref{thm:sc-lower}]
	Fix a $p$-pass $s$-space streaming algorithm $\alg$ for the set cover problem over random arrival streams that outputs an $\alpha$-approximation w.p. at least $1-\delta$ for $\delta < 1/4$. One
	 can easily turn $\alg$ into a $\delta$-error protocol for $\SetCover$ on the distribution $\distSCrnd$: Alice and Bob take a random permutation of their inputs and then treat their combined input 
	 as a set stream and run $\alg$ on that. The random partitioning of the input plus the random permutation taken by the players ensure that the constructed stream is a random permutation of the input sets. Consequently,
	 this protocol is a $\delta$-error protocol for $\SetCover$ on $\distSCrnd$ that uses $O(p \cdot s)$ bits of communication. Since $\delta < 1/4$, by Lemma~\ref{lem:dist-random}, $p \cdot s = \Omgt(mn^{\frac{1}{\alpha}})$, 
	 proving the theorem. 
\end{proof}

\newcommand{\sol}{\ensuremath{\textnormal{\textsf{SOL}}\xspace}}
\newcommand{\Usmpl}{\ensuremath{U_{\textnormal{\textsf{smpl}}}}}

\subsection{An $\alpha$-Approximation Algorithm for the Streaming Set Cover Problem} \label{sec:sc-upper} 

In this section, we prove the optimality of the lower bound in Theorem~\ref{thm:sc-lower} by establishing a matching upper bound (i.e. Theorem~\ref{thm:sc-upper}). 
As stated earlier, our algorithm is a simple modification of the algorithm of~\cite{HarPeledIMV16}. In particular, we obtain our improved algorithm by using a one-shot pruning step 
as opposed to the iterative pruning of~\cite{HarPeledIMV16}, and employing a more careful element sampling (compare the bounds in Lemma~\ref{lem:element-sampling} in this paper
with Lemma~2.5 in~\cite{HarPeledIMV16}). 

In the following, we assume that we are given a value $\topt$ which is a $(1+\eps)$-approximation of $\opt$, i.e., 
the optimal solution size of the given instance. This is without loss of generality as we can run the algorithm in parallel for $O(\log{n}/\eps)$ guesses for $\topt \in [1,n]$ and return the 
smallest computed set cover among all parallel runs. 

The general idea behind the algorithm is as follows: we know that $\topt$ sets are enough to cover the whole universe $[n]$; hence, if we find 
a $(1-\rho)$-approximate $k$-cover of the input sets for the parameter $k = \topt$ and $\rho = 1/n^{1/\alpha}$, we can reduce the number of uncovered elements by a factor of $n^{1/\alpha}$.  
Repeating this process $\alpha$ times then results in a collection of at most $\alpha \cdot \topt$ sets that covers the whole universe, i.e., an $\alpha$-approximate set cover. It is worth mentioning 
that this is the general principle behind most (but not all) streaming algorithms for set cover, see, e.g.~\cite{HarPeledIMV16,BateniEM16,SahaG09,DemaineIMV14}. 

Notice that we can readily use the maximum coverage streaming algorithms of~\cite{McGregorVu16,BateniEM16} as a sub-routine to find the approximate $k$-cover above; however, doing so would result in a sub-optimal algorithm for set cover as these algorithms have space dependence of (at least) $\Omega(m/\rho^2) = \Omega(mn^{2/\alpha})$ (even ignoring the dependence on $k$, i.e., $\topt$). In fact, as we prove in the next section (see Result~\ref{res:kc-lower}), any $(1-\rho)$-approximate $k$-cover algorithm needs $\Omega(m/\rho^2)$ space in general. To bypass this, we crucially use the fact that the aforementioned maximum coverage instances 
have the additional property that the optimal answer is the whole universe and hence the element sampling technique of~\cite{HarPeledIMV16} (and similar ones in~\cite{McGregorVu16,BateniEM16}) can be improved 
for this special case. We now provide the formal description of the algorithm. 

\newcounter{algo}

\textbox{Algorithm~\refstepcounter{algo}\thealgo\label{alg1}. \textnormal{An $\alpha$-approximation algorithm for the streaming set cover problem.}}{

\medskip
\textbf{Input.} A stream $\SS=(S_1,\ldots,S_m)$ of subsets of $[n]$, and a $(1+\eps)$-approximation $\topt$ of $\opt(\SS)$. \\
\textbf{Output.} A collection of $(1+\eps)\cdot\alpha\cdot\topt$ sets that cover the universe. \\
\algline

\begin{enumerate}
	\item Let $U \leftarrow [n]$ and $\sol \leftarrow \emptyset$. 
	\item Make a single pass over the stream and if $\card{S_i \cap U} \geq {n}/{(\eps \cdot \topt)}$, then: 
	\begin{enumerate}
		\item $\sol \leftarrow \sol \cup \set{i}$ and $U \leftarrow U \setminus S_i$. 
	\end{enumerate}
	\item For $j=1$ to $\alpha$ iterations: 
	\begin{enumerate}
		\item Let $\Usmpl$ be a subset of $U$ chosen by picking each element independently and w.p. $p = 16\cdot \topt\cdot \log{m} / n^{1-1/\alpha}$. 
		\item Make a single pass over the stream and for all $i \in [m]$, store $S'_i = S_i \cap \Usmpl$ in the memory. 
		\item Find an optimal set cover $\OPT'$ of the instance $(S'_1,\ldots,S'_m)$ and let $\sol \leftarrow \sol \cup \OPT'$. 
		\item Make another pass over the stream and let $\Usmpl \leftarrow \Usmpl \setminus \bigcup_{i \in \OPT'} S_i$. 
	\end{enumerate}
	\item Return $\sol$ as a set cover of the input instance. 
\end{enumerate}
}

We start by bounding the space requirement of Algorithm~\ref{alg1} .

\begin{lemma}\label{lem:space-requirement}
	Algorithm~\ref{alg1} requires $\Ot(mn^{1/\alpha}/\eps + n)$ space w.p. at least $1-1/m^2$. 
\end{lemma}
\begin{proof}
	It is easy to see that maintaining $\sol$ and $U$ requires, respectively, $O(m)$ and $O(n)$ space. In the following, we analyze the space required for storing the sets
	$(S'_1,\ldots,S'_m)$. After the first pass of the algorithm, no set contains more than $n/(\eps \cdot \topt)$ elements in $U$. 
	Fix a set $S_i \in \SS$; we have, 
	\begin{align*}
		\EX\card{S_i \cap \Usmpl} &= \card{S_i} \cdot p \leq n/(\eps \cdot \topt) \cdot \paren{16\cdot \topt\cdot \log{m} / n^{1-1/\alpha}} \\
		&=  16 \cdot n^{1/\alpha} \cdot \log{m} / \eps
	\end{align*}
	Hence, by Chernoff bound, w.p. $1-1/m^{3}$, $\card{S_i \cap \Usmpl} = \Ot(n^{1/\alpha}/\eps)$. The final bound now follows from this and a union bound on all $m$ sets in $\SS$. 
\end{proof}

\begin{remark}\label{rem:make-deterministic}
One can make the space requirement of Algorithm~\ref{alg1} deterministic by terminating the algorithm whenever it attempts to use a memory more than the bounds in Lemma~\ref{lem:space-requirement}. As 
this event happens with negligible probability, the correctness of the algorithm can be argued exactly the same. 
\end{remark}

The following two lemmas establish the correctness of the algorithm. 

\begin{lemma}\label{lem:apx-guarantee}
	Algorithm~\ref{alg1} picks at most $\paren{\alpha + \eps}  \cdot \topt$ sets in $\sol$. 
\end{lemma}
\begin{proof}
	It is immediate to see that in the first pass, the algorithm picks at most $\eps \cdot \topt$ sets as otherwise $U$ would be empty. 
	Moreover, in each subsequent $\alpha$ iterations, the algorithm picks at most $\topt$ sets since $(S'_1,\ldots,S'_m)$ has a
	set cover of size at most $\topt$ (as the original instance had a set cover of size $\leq \topt$). 
\end{proof}

\begin{lemma}\label{lem:feasible}
	The set $\sol$ computed by Algorithm~\ref{alg1} is a feasible set cover of $[n]$ w.p. $1-1/m$. 
\end{lemma}

To prove Lemma~\ref{lem:feasible}, we use the following property of element sampling that first appeared in~\cite{DemaineIMV14} (similar ideas also appear in~\cite{McGregorVu16,HarPeledIMV16}); 
for completeness we provide a self-contained proof of this lemma here. 

\begin{lemma}\label{lem:element-sampling}
	Let $0 < \rho < 1$ be a parameter and $\SS = (S_1,\ldots,S_m)$ be a collection of $m$ subsets of $[n]$ with $\opt(\SS) \leq k$. Suppose $\Usmpl$ is a subset of $[n]$ obtained by picking
	each element independently and w.p. $p \geq 16  \cdot k \cdot \log{m}/(\rho \cdot n)$; then, w.p. $1-1/m^2$, any collection of $k$ sets in $\SS$ 
	that covers $\Usmpl$ entirely also covers at least $\paren{1-\rho} \cdot n$ elements in $[n]$. 
\end{lemma}
\begin{proof}
	Fix a collection $\SC$ of $k$ subsets in $\SS$ that covers less than $(1-\rho) \cdot n$ elements in $[n]$. The probability that this collection covers $\Usmpl$ entirely is equal to the probability that none of the 
	$\rho \cdot n$ elements in $[n]$ that are not appearing in $\SC$ are sampled in $\Usmpl$. Hence, 
	\begin{align*}
		\Pr{\text{$\SC$ covers $\Usmpl$}} &\leq \paren{1-p}^{\rho \cdot n} \leq \exp\paren{-\paren{16 \cdot k \cdot \log{m}/(\rho \cdot n)} \cdot (\rho \cdot n)} \leq 1/m^{8k}
	\end{align*}
	Taking a union bound over all ${{m} \choose {k}} \leq m^{k}$ possible choices for $\SC$ finalizes the result. 
\end{proof}
\begin{proof}[Proof of Lemma~\ref{lem:feasible}]
	In each of the $\alpha$ iterations, Algorithm~\ref{alg1} implements the sampling in Lemma~\ref{lem:element-sampling} with the parameters $k = \topt$, and $\rho = n^{-1/\alpha}$. Hence, after 
	each iteration, the number of uncovered elements in $U$ reduces to $\card{U}/n^{1/\alpha}$ w.p. $1-1/m^{2}$. Consequently, by taking a union bound over the $\alpha \leq m$ iterations, 
	after the $\alpha$ iterations, number of uncovered elements reduces to less than $1$, hence proving the lemma. 
\end{proof}

We now conclude the proof of Theorem~\ref{thm:sc-upper}. 
\begin{proof}[Proof of Theorem~\ref{thm:sc-upper}]
	We can run Algorithm~\ref{alg1} in 
	parallel for $O(\log{n}/\eps)$ possible guesses for $\topt$. By Lemma~\ref{lem:space-requirement}, the space requirement of this algorithm is $\Ot(1/\eps) \cdot \Ot(mn^{1/\alpha}/\eps + n)$ as desired.
	Moreover, consider the guess: $\opt \leq \topt \leq (1+\eps)\cdot\topt$. For this choice, we can apply Lemma~\ref{lem:apx-guarantee} and Lemma~\ref{lem:feasible} and obtain that the returned solution is an 
	$(\alpha + O(\eps))$-approximation of the optimal set cover. Since the algorithm can make sure that the returned solution is always feasible, returning the smallest set cover among all guesses for $\topt$ 
	then ensures that the returned answer is an $(\alpha + O(\eps))$ approximation. Re-parameterizing $\eps$ by a constant factor, finalizes the proof. 
\end{proof}

\newcommand{\GHD}{\ensuremath{\textnormal{\textsf{GHD}}}\xspace}
\newcommand{\MaxCover}{\ensuremath{\textnormal{\textsf{MaxCover}}}\xspace}

\newcommand{\distU}{\ensuremath{\mathcal{U}}}

\newcommand{\distGHDY}{\ensuremath{\dist}_{\GHD}^{\textnormal{\textsf{Y}}}}
\newcommand{\distGHDN}{\ensuremath{\dist}_{\GHD}^{\textnormal{\textsf{N}}}}
\newcommand{\distGHD}{\ensuremath{\dist}_{\GHD}}

\newcommand{\protGHD}{\ensuremath{\prot}_{\GHD}}

\section{The Space-Approximation Tradeoff for Maximum Coverage} \label{SEC:K-COVER}

In this section, we prove a space-approximation tradeoff for the maximum coverage problem.

\begin{theorem}\label{thm:kc-lower}
  For any $\eps = \omega(1/\sqrt{n})$, $m = \poly(n)$, and $p \geq 1$, any \emph{randomized} algorithm that can make $p$ passes over any collection of $m$ subsets of $[n]$ presented in a 
  {random order stream} and outputs a $(1-\eps)$-approximation to the optimal value of the maximum coverage problem for $k=2$ with a sufficiently large constant probability (over the randomness of both the
  stream order and the algorithm) must use $\Omgt(m/(\eps^2 \cdot p))$ space. 
\end{theorem}

Similar to previous section, we prove Theorem~\ref{thm:kc-lower} by considering the communication complexity of the maximum coverage problem: Fix a (sufficiently large) $n$, $\eps = \omega(1/\sqrt{n})$ and
$m = \poly(n)$; $\MaxCover$ refers to the communication problem of $(1-\eps)$-approximating the optimal value of the maximum coverage problem with $2m$ sets defined over the universe $[n]$ and parameter $k=2$, in the two-player communication model.

Our lower bound for $\MaxCover$ is obtained by reducing this problem to multiple instances of the \emph{gap-hamming-distance} problem via a similar distribution as $\distSC$ (using an additional simple gadget). 
In the following, we first introduce the gap-hamming-distance problem and prove a useful lemma on its information complexity on particular distributions required for our reduction, and then describe a hard distribution
for \MaxCover based on this and finalize the proof of Theorem~\ref{thm:kc-lower}. 

\subsection{The Gap-Hamming-Distance Problem}\label{sec:ghd}

The \emph{gap-hamming-distance} ($\GHD$) problem is defined as follows. Fix an integer $t \geq 1$; in $\GHD_{t}$, Alice is given a set $A \subseteq [t]$, Bob 
is given a set $B \subseteq [t]$ and the goal is to output: 
\begin{align*}
	\GHD(A,B) = 
	\begin{cases} 
	\Yes~~~&\Delta(A,B) \geq t/2 + \sqrt{t} \\
	\No~~~&\Delta(A,B) \leq t/2 - \sqrt{t} \\
	\star~~~&\textnormal{otherwise}
	\end{cases}
\end{align*}
where $\star$ means that the answer can be arbitrary; here $\Delta(A,B)$ denotes the hamming distance between $A$ and $B$, i.e., the size of the symmetric difference of $A$ and $B$. 

This problem was originally introduced by~\cite{IndykW03} and has been studied extensively in the literature (see~\cite{ChakrabartiR11} and 
references therein). We use the following result on the information complexity of this problem proven in~\cite{BravermanGPW13sr}\footnote{Technically speaking,~\cite{BravermanGPW13sr} bounds
the \emph{external} information complexity of $\GHD$ as opposed to its \emph{internal} information complexity used in our paper. However, since the distribution in Lemma~\ref{lem:ic-ghd-uniform} is a product 
distribution, these two quantities are equal and hence we simply state the bound for the internal information complexity.}.
\begin{lemma}[\!\!\cite{BravermanGPW13sr}]\label{lem:ic-ghd-uniform}
	Let $\distU$ be the uniform distribution on pairs of subsets of $[t]$ (chosen independently); there exists an absolute
	 constant $\delta > 0$ such that \[\IC{\GHD}{\distU}{\delta} = \Omega(t).\]
\end{lemma}

For our purpose, we need to consider the following distribution $\distGHD$ for $\GHD$ instead of the uniform distribution. 
Let $a,b \in [t]$ be two parameters to be determined later\footnote{The values of $a$ and $b$ are not important for our purpose and are hence only determined in
the proof of Lemma~\ref{lem:ic-ghd}.}. Define: 
\begin{itemize}
	\item $\distGHDY$ as the distribution of instances $(A,B) \sim \distU \mid \Delta(A,B) \geq t/2 + \sqrt{t}, \card{A} = a,\card{B} = b$.
	\item $\distGHDN$ as the distribution of instances $(A,B) \sim \distU \mid \Delta(A,B) \leq t/2 - \sqrt{t}, \card{A} = a, \card{B} = b$. 
	\item $\distGHD:= \frac{1}{2} \cdot \distGHDY + \frac{1}{2} \cdot \distGHDN$.
\end{itemize}

We use Lemma~\ref{lem:ic-ghd-uniform} to prove the following result on the information cost of $\delta$-error protocols on the distribution $\distGHD$,
which could be independently useful also. The proof is deferred to Appendix~\ref{app:ic-ghd}. 

\begin{lemma}\label{lem:ic-ghd}
	Let $\delta > 0$ be a sufficiently small constant and $\protGHD$ be a $\delta$-error protocol for $\GHD_t$ on $\distGHD$ with $\norm{\protGHD} = 2^{o(t)}$; then, $\ICost{\protGHD}{\distGHDN} = \Omega(t)$. 
\end{lemma}

\subsection{Communication Complexity of \MaxCover}\label{sec:cc-k-cover}

We are now ready to prove a lower bound on the communication complexity of the \MaxCover problem. To do so, we propose the
following distribution. 

\newcommand{\distKC}{\ensuremath{\dist_{\textnormal{\textsf{MC}}}}}
\newcommand{\protKC}{\ensuremath{\prot_{\textnormal{\textsf{MC}}}}}

\textbox{Distribution \distKC. {\textnormal{A hard input distribution for $\MaxCover$.}}} {

\medskip

\textbf{Notation.} Let $t_1:= 1/\eps^2$, $t_2:= 10 \cdot t_1$, $U_1:= [t_1]$ and $U_2 := [t_1+1,t_1+t_2]$.  
\begin{itemize}
	\item For each $i \in [m]$: 
	\begin{itemize}
		\item Let $(A_i,B_i) \sim \distGHDN$ for $\GHD_{t_1}$ on the universe $U_1$. 
		\item Create $C_i,D_i \subseteq U_2$, by assigning each element in $U_2$ w.p. $1/2$ to $C_i$ and o.w. to $D_i$. 	
		\item Let $S_i := A_i \cup C_i$ and $T_i := B_i \cup D_i$. 
	\end{itemize}
	\item Pick $\theta \in_R \set{0,1}$ uniformly at random. If $\theta = 0$, do nothing, otherwise: 
	\begin{itemize}
		\item Sample $\istar \in_R [m]$ uniformly at random.
		\item Resample $(A_{\istar},B_{\istar}) \sim \distGHDY$ for $\GHD_{t_1}$ and redefine $S_{\istar}$ and $T_{\istar}$ as before using the new pair $(A_{\istar},B_{\istar})$ (do not change $C_i$ and $D_i$). 
	\end{itemize}
	\item Let the input to Alice and Bob be $\SS:=\set{S_i}_{i\in[m]}$ and $\ST:=\set{T_i}_{i\in[m]}$, respectively. 
\end{itemize}
}

Define $\opt(\SS,\ST)$ as the value of the optimal solution of the maximum coverage problem (for the parameter $k=2$) for the instance $(\SS,\ST)$. We wish to argue that 
$\opt(\SS,\ST)$ differs by a $(1\pm \eps)$ factor depending on the choice of $\theta$ in the distribution and hence any $(1-\eps)$ approximation algorithm for maximum coverage on this distribution
needs to determine the value of $\theta$.  
\begin{lemma}\label{lem:opt-kc}
	Assuming $\eps = o(1/\log{n})$, there exists a fixed $\tau \in [n]$ such that for any instance $(\SS,\ST) \sim \distKC$: 
	\begin{align*}
		&\Pr{\opt(\SS,\ST) \geq (1+\Theta(\eps)) \cdot \tau \mid \theta = 1} = 1-o(1) \\
		&\Pr{\opt(\SS,\ST) \leq (1-\Theta(\eps)) \cdot \tau \mid \theta = 0} = 1-o(1) 
	\end{align*}
\end{lemma}
\begin{proof}
	We first prove that, any $(1-\eps)$-approximate $2$-cover in this distribution always has to pick a pair of $(S_i,T_i)$ sets (for some $i \in [m]$). This is achieved by considering 
	the projection of the sets on the universe $U_2$. 
	\begin{claim}\label{clm:kc-dominant-strategy}
		W.p. $1-o(1)$: 
		\begin{enumerate}[(a)]
			\item For any $i \in [m]$, $\card{S_i \cup T_i} \geq t_2$.  
			\item For any $i \neq j \in [m]$, for any $Z_i \in \set{S_i,T_i}$, and $Z_j \in \set{S_j,T_j}$, $\card{Z_i \cup Z_j} \leq (3/4+0.2) \cdot t_2$.
		\end{enumerate}
	\end{claim}
	\begin{proof}
		Part~$(a)$ follows immediately from the fact that $U_2$ is partitioned between $S_i$ and $T_i$, and that $\card{U_2} = t_2$. We now prove Part~$(b)$. To do so, we prove
		that $Z_i \cup Z_j$ can only cover (essentially) $3/4$ fraction of $U_2$ w.h.p and since the rest of $Z_i \cup Z_j$ is a subset of $U_1$ with $\card{U_1} \leq 0.1 \cdot t_2$, we get the final result. 
		
		For any element $e \in U_2$, define an indicator random variable $X_e \in \set{0,1}$ whereby $X_e = 1$ iff $e \in Z_i \cup Z_j$. Since $i \neq j$, the elements in $Z_i$ and $Z_j$ that are in $U_2$ are chosen
		independent of each other, and hence $\Pr{X_e = 1} = 1 - \paren{1/2}^{2} = 3/4$. Define $X := \sum_{e \in U_2} X_e$; we have $\Ex{X} = 3/4 \cdot t_2$ and since $X_e$ variables are independent, by Chernoff
		bound, $\Pr{X \geq \Ex{X} + 0.1 \cdot t_2} \leq \exp\paren{-c \cdot t_2} = o(1/m^2)$ (as $t_2 = \omega(\log{n})$ and $m = \poly{(n)}$).  The final result now follows from a union bound on all possible ($\leq (2m)^2$) pairs. 
	\end{proof}
	Now consider a pair $(S_i,T_i)$ for some $i \in [m]$ and note that $\card{S_i \cup T_i} = \card{U_2} + \card{A_i \cup B_i} = t_2 + \card{A_i \cup B_i}$; 
	hence we can simply focus on $A_i \cup B_i \subseteq U_1$ part of $S_i \cup T_i$. Moreover, we have that, 
	\begin{align*}
	\card{A_i \cup B_i} &= \frac{1}{2} \cdot \paren{\card{A_i} + \card{B_i} + \Delta(A_i,B_i)} = \frac{1}{2} \cdot \paren{a+b + \Delta(A_i,B_i)}
	\end{align*}
	where we used the fact that in the distribution $\distGHD$, $\card{A_i} =a$ and $\card{B_i} =b$ always. 
	
	Consequently, whenever $(A_i,B_i) \sim \distGHDN$, we have, 
	\begin{align*}
		\card{S_i \cup T_i} &= t_2 + \card{A_i \cup B_i} \leq t_2 + (a+b)/2 +  {t_1}/{4} - \sqrt{t_1}/{2} = (1-\Theta(\eps)) \cdot \tau
	\end{align*}
	for $\tau := t_2 + (a+b)/2 +  {t_1}/{4}$. Similarly, whenever $(A_i,B_i) \sim \distGHDY$, 
	\begin{align*}
		\card{S_i \cup T_i} &\geq t_2 + \card{A_i \cup B_i} \geq t_2 + (a+b)/2 +  {t_1}/{4} + \sqrt{t_1}/{2} = (1+\Theta(\eps)) \cdot \tau
	\end{align*}
	Combining these bounds with Claim~\ref{clm:kc-dominant-strategy} finalizes the proof. 
\end{proof}

Having proved Lemma~\ref{lem:opt-kc}, we can use any $(1-\eps)$-approximation protocol for $\MaxCover$ to determine the parameter $\theta$ in the distribution $\distKC$ (by a simple re-parametrizing of the $\eps$ by 
a constant factor). This allows us to prove the following lemma. The proof is essentially identical to that of Lemma~\ref{lem:direct-sum} in Section~\ref{sec:sc-lower} and is provided only for the sake 
of completeness. 

\begin{lemma}\label{lem:kc-direct-sum}
	Let $\protKC$ be a $\delta$-error protocol for $\MaxCover$ on $\distKC$. There exists a $\paren{\delta+o(1)}$-protocol $\protGHD$ for $\GHD_{t_1}$ on the distribution $\distGHD$ such that:
	\begin{enumerate}
	\item  $\ICost{\protGHD}{\distGHDN} = \frac{O(1)}{m} \cdot \ICost{\protKC}{\distKC}$.
	\item $\norm{\protGHD} = \norm{\protKC}$. 
	\end{enumerate}
\end{lemma}

\newcommand{\ProtGHD}{\ensuremath{\Prot}_{\GHD}}
\newcommand{\ProtKC}{\ensuremath{\Prot}_{\textnormal{\textsf{MC}}}}
	
\renewcommand{\bC}{\bm{C}}
\renewcommand{\bD}{\bm{D}}

\begin{proof}
	We design the protocol $\protGHD$ as follows: 
	\textbox{Protocol $\protGHD$. \textnormal{The protocol for solving $\GHD_{t_1}$ using a protocol $\protKC$ for $\MaxCover$.}}{
	\medskip
	\\ 
	\textbf{Input:} An instance $(A,B) \sim \distGHD$.  \textbf{Output:} $\GHD(A,B)$. \\
	\algline
	
	\begin{enumerate}
		\item Using public randomness, the players sample an index $\istar \in_R [m]$.
		\item Using public randomness, the players sample the sets $A^{<\istar}$ and $B^{>\istar}$ each from $\distGHDN$ independently. 
		\item Using private randomness, Alice samples the sets $A^{>\istar}$ such that $(A_j,B_j) \sim \distGHDN$ (for all $j > \istar$); similarly Bob samples the sets $B^{<\istar}$. 
		\item Using public randomness the players sample the sets $(C_i,D_i)$ for all $i \in [m]$ from a (distinct) universe $U_2$ the same as distribution $\distKC$. 
		\item The players construct the collections $\SS:= \set{S_1,\ldots,S_m}$ and $\ST:= \set{T_1,\ldots,T_m}$ by setting $S_i := A_i \cup C_i$ and $T_i:= B_i \cup D_i$ (exactly as in
		distribution $\distKC$).
		\item The players solve the \MaxCover instance using \protKC~and output \No iff $\protKC$ estimates $\opt(\SS,\ST) \leq \tau$ (for the parameter $\tau$ in Lemma~\ref{lem:opt-kc}) and \Yes otherwise. 
	\end{enumerate}
	}

	It is easy to see that the distribution of instances $(\SS,\ST)$ created in the protocol $\protKC$ matches the distribution $\distKC$ for $\MaxCover$ exactly, and hence by Lemma~\ref{lem:opt-kc}, 
	$\protGHD$ is a $(\delta+o(1))$-error protocol for $\GHD$ in the distribution $\distKC$.  The bound on the communication cost of $\protKC$ is also immediate; in the following we bound the 
	information cost of $\protKC$ for $(A,B)$ sampled from $\distGHDN$. 
	
	Let $I$ be a random variable for $\istar$ and $R$ be the set of public randomness used by the players. By Claim~\ref{clm:public-random}, 
	\begin{align*}
		\ICost{\protGHD}{\distGHDN} &= \II_{\distGHDN}(\ProtGHD : A \mid B,R) + \II_{\distGHDN}(\ProtGHD: B \mid A,R)
	\end{align*}
	We now bound the first term in the RHS above (the second term can be bounded exactly the same). In the following, let $\bC$ and $\bD$ denote the vector of random variables for $C_i$'s and $D_i$'s, 
	respectively. 
	\begin{align*}
		\II_{\distGHDN}(\ProtGHD : A \mid B,R) &= {\II_{\distGHDN}(\ProtGHD : A \mid B,R,I)} \tag{$I$ is chosen using public randomness} \\
		&= \EX_{i} \Bracket{\II_{\distGHDN}\paren{\ProtGHD : A_i \mid B_i,A^{<i},B^{>i},\bC,\bD,I=i} }\tag{$R = (A^{<i},B^{>i},\bC,\bD,I)$} \\ 
		&= \sum_{i=1}^{m} \frac{1}{m} \cdot \II_{\distGHDN}\paren{\ProtGHD : A_i \mid B_i,A^{<i},B^{>i},\bC,\bD} 		
	\end{align*}
	where the last equality is because conditioned on $(A,B) \sim \distGHDN$, all sets $A_j,B_j$ (for $j \in [m]$) are chosen from $\distGHDN$ and hence are independent of the $``I=i"$ event.
	We can further derive, 
	\begin{align*}
		\II_{\distGHDN}(\ProtGHD : A \mid B,R) &= \sum_{i=1}^{m} \frac{1}{m} \cdot \II_{\distGHDN}\paren{\ProtGHD : A_i \mid B_i,A^{<i},B^{>i},\bC,\bD} \\
		&\leq \frac{1}{m} \cdot \sum_{i=1}^{m} \II_{\distGHDN}\paren{\ProtGHD : A_i \mid A^{<i},\bB,\bC,\bD} \tag{$A_i \perp B^{<i} \mid \bB,\bC,\bD$ and hence we can apply Fact~\ref{fact:info-increase}} \\
		&= \frac{1}{m} \cdot \II_{\distGHDN}\paren{\ProtGHD :  \bA \mid \bB,\bC,\bD} \tag{chain rule of mutual information, \itfacts{chain-rule}} \\
		&= \frac{1}{m} \cdot \II_{\distGHDN}\paren{\ProtGHD :  \SS \mid \ST,\bC,\bD} \\
		&= \frac{1}{m} \cdot \II_{\distKC}\paren{\ProtKC :  \SS \mid \ST,\bC,\bD,\theta=0} 
	\end{align*}
	where the second last equality is because $\bA$ (resp. $\bB$) and $\SS$ (resp. $\ST$) determine each other conditioned on $\bC$ and $\bD$, and 
	last equality is because the distribution of maximum coverage instances and the messages communicated by the players under $\distGHDN$ and under $\distKC \mid \theta = 0$ exactly matches. 
	
	Moreover, 
	\begin{align*}
		\II_{\distGHDN}(\ProtGHD : A \mid B,R) &\leq \frac{1}{m} \cdot \II_{\distKC}\paren{\ProtKC :  \SS \mid \ST,\bC,\bD,\theta=0} \\
		&\leq \frac{2}{m} \cdot \II_{\distKC}\paren{\ProtKC :  \SS \mid \ST,\bC,\bD,\theta} \tag{by definition of mutual information as $\Pr{\theta = 0} = 1/2$} \\
		&\leq \frac{2}{m} \cdot \paren{\II_{\distKC}\paren{\ProtKC :  \SS \mid \ST,\bC,\bD} + H(\theta)} \tag{by Fact~\ref{fact:bar-hopping}}\\
		&= \frac{2}{m} \cdot \II_{\distKC}\paren{\ProtKC : \SS \mid \ST,\bC,\bD} + \frac{2}{m} \tag{$H(\theta) = 1$ by \itfacts{uniform}}  \\
		&\leq \frac{2}{m} \cdot \II_{\distKC}\paren{\ProtKC :  \SS \mid \ST} + \frac{2}{m} \tag{$\ProtKC \perp \bC,\bD \mid \SS,\ST$ and hence we can apply Fact~\ref{fact:info-decrease}}
	\end{align*}
	By performing the same exact calculation for $\II_{\distGHDN}(\ProtGHD: B \mid A,R)$, we obtain that, 
	\begin{align*}
		\ICost{\protGHD}{\distGHDN} &\leq \frac{2}{m} \cdot \paren{ \II_{\distKC}\paren{\ProtKC : \SS \mid \ST}  +  \II_{\distKC}\paren{\ProtSC : \ST \mid \SS} } + \frac{4}{m} \\
		&= \frac{2}{m} \cdot \ICost{\protKC}{\distKC} + \frac{4}{m} = \frac{O(1)}{m} \cdot \ICost{\protKC}{\distKC} 
	\end{align*}
	where in the last inequality we used the fact that information cost of $\protKC$ is at least $1$. 
\end{proof}

We now have, 

\begin{theorem}\label{thm:kc-cc-lower}
	There exists a sufficiently small constant $\delta > 0$, such that for any $\omega(1/\sqrt{n}) \leq \eps \leq o(1/\log{n})$, and $m = \poly{(n)}$, 
	\begin{align*}
		\CC{\MaxCover}{\distKC}{\delta} = \Omega({m/\eps^2}).
	\end{align*} 
\end{theorem}
\begin{proof}
	Suppose there exists a $\delta$-error protocol $\protKC$ for $\MaxCover$ on $\distKC$ for a sufficiently small constant $\delta$ (to be determined later), with $\norm{\protGHD} = o(m/\eps^2)$; 
	by Proposition~\ref{prop:cc-ic}, $\ICost{\protKC}{\distKC} = o(m/\eps^2)$ as well. Hence, by Lemma~\ref{lem:kc-direct-sum}, we obtain a $(\delta+o(1))$-error protocol $\protGHD$ for $\GHD_{t_1}$ on $\distGHD$ 
	with $\norm{\protGHD} = o(m/\eps^2)$ and 
	$\ICost{\protGHD}{\distGHDN} = o(1/\eps^2) = o(t_1)$.  However, since $\norm{\protGHD} = o(m/\eps^2) = 2^{o(t_1)}$ as $m = \poly(n)$ and $t_1 = \omega(\log{n})$, we can
	 now apply Lemma~\ref{lem:ic-ghd} and argue that $\ICost{\protGHD}{\distGHDN}$ is $\Omega(t_1)$ (by taking $\delta$ smaller than the bounds in the Lemma~\ref{lem:ic-ghd}); a contradiction
	 with the information cost of $\protGHD$ obtained by Lemma~\ref{lem:kc-direct-sum}. 
\end{proof}

 We point out that to extend the results in Theorem~\ref{thm:kc-cc-lower} to $\eps > 1/\log{n}$ case (i.e., the case not handled by Theorem~\ref{thm:kc-cc-lower}), we can simply use an existing $\Omgt(m)$ lower
 bound of~\cite{McGregorVu16} (Theorem 21) for this range of the parameter $\eps$. 
 
We can now prove Theorem~\ref{thm:kc-lower} by using Theorem~\ref{thm:kc-cc-lower}, the same exact way as we proved Theorem~\ref{thm:sc-lower}, i.e., by defining a random partitioning
 version of the distribution $\distKC$ and proving the lower bound using that partitioning. We briefly sketch the proof here.  
 
 \begin{proof}[Proof Sketch of Theorem~\ref{thm:kc-lower}] 
	Define the distribution $\distKC'$ similar to the distribution $\distKC$ with the difference that after creating the sets $\SS$ and $\ST$, 
	we randomly partition the sets between the players (i.e., assign each set to Alice w.p. $1/2$ and o.w. to Bob). The same exact argument in Lemma~\ref{lem:dist-random}, combined with Theorem~\ref{thm:kc-cc-lower}
	(instead of Theorem~\ref{thm:sc-cc-lower} in Lemma~\ref{lem:dist-random}) now proves that for some sufficiently small constant $\delta > 0$, $\CC{\MaxCover}{\distKC'}{\delta} = \Omega(m/\eps^2)$. 
	
	Furthermore, any $p$-pass $s$-space streaming algorithm for maximum coverage on random arrival streams can be turned into an $O(s \cdot p)$-bit communication
	protocol for $\MaxCover$ on $\distKC'$ (with the same error probability); see the proof of Theorem~\ref{thm:sc-lower} for more details. This, together with the lower bound on the distribution 
	$\distKC'$ implies that $s = \Omega(m/(\eps^2 \cdot p))$ as desired.  
\end{proof}

\subsection*{Acknowledgements}
I am grateful to my advisor Sanjeev Khanna for valuable discussions, and to Ehsan Emamjomeh-Zadeh and Sanjeev Khanna for carefully reading the paper and many helpful comments. 
I also thank the anonymous reviewers of PODS 2017 for many insightful comments and suggestions.

\bibliographystyle{acm}
\bibliography{general}

\clearpage
\appendix
\section{Tools from Information Theory}\label{app:info}

Here, we briefly introduce some basic facts from information theory that are needed in this paper. We refer the interested reader to the textbook by Cover and Thomas~\cite{ITbook} for an excellent introduction to this field. 

We use the following basic properties of entropy and mutual information (proofs can be
found in~\cite{ITbook}, Chapter~2).
\begin{fact}\label{fact:it-facts}
  Let $ A$, $ B$, and $ C$ be three (possibly correlated) random variables.
   \begin{enumerate}
  \item \label{part:uniform} $0 \leq \HH( A) \leq \card{ A}$. $\HH( A) = \card{ A}$
    iff $ A$ is uniformly distributed over its support.
  \item \label{part:info-zero} $\II( A :  B) \geq 0$. The equality holds iff $ A$ and
    $ B$ are \emph{independent}.
  \item \label{part:cond-reduce} \emph{Conditioning on a random variable reduces entropy}:
    $\HH( A \mid  B, C) \leq \HH( A \mid  B)$.  The equality holds iff $ A \perp C \mid B$.
  \item \label{part:chain-rule} \emph{The chain rule for mutual information}: $\II( A, B :  C) = \II( A :  C) + \II(  B: C \mid  A)$.
   \end{enumerate}
\end{fact}

We also use the following two simple facts, which assert conditions in which conditioning can provably increase (resp. decrease) the mutual information. 

\begin{fact}\label{fact:info-increase}
  For random variables $ A,  B,  C,D$, if $A\perp D \mid C$, then $\II( A :  B \mid  C) \leq \II( A :  B \mid  C,  D)$.
\end{fact}
 \begin{proof}
  Since $ A$ and $ D$ are independent conditioned on $C$, by
  \itfacts{cond-reduce}, $\HH( A \mid  C) = \HH( A \mid
   C, D)$ and $\HH( A \mid  C, B) \ge \HH( A \mid  C, B, D)$.  We have,
	 \begin{align*}
	  \II( A :  B \mid  C) &= \HH( A \mid  C) - \HH( A \mid  C, B) = \HH( A \mid  C, D) - \HH( A \mid  C, B) \\
	  &\leq \HH( A \mid  C, D) - \HH( A \mid  C, B, D) = \II( A : B \mid  C, D)
	\end{align*}
\end{proof}

\begin{fact}\label{fact:info-decrease}
  For random variables $ A,  B,  C,D$, if $ A \perp D \mid B,C$, then, $\II(A : B \mid  C) \geq \II(A : B \mid  C,  D)$.
\end{fact}
 \begin{proof}
 Since $A \perp D \mid B,C$, by \itfacts{cond-reduce}, $\HH(A \mid B,C) = \HH(A \mid B,C,D)$. Moreover, since conditioning can only reduce the entropy (again by \itfacts{cond-reduce}), 
  \begin{align*}
 	\II( A :  B \mid  C) &= \HH(A \mid C) - \HH(A \mid B,C) \geq \HH(A \mid D,C) - \HH(A \mid B,C) \\
	&= \HH(A \mid D,C) - \HH(A \mid B,C,D) = \II(A : B \mid C,B) 
 \end{align*}
\end{proof}

Finally, we use the following simple inequality that states that conditioning on a random variable can only increase the mutual information
by the entropy of the conditioned variable. 

\begin{fact}\label{fact:bar-hopping}
  For any random variables $ A,  B$ and $C$, $\II(A : B \mid C) \leq \II(A : B) + \HH(C)$. 
\end{fact}
\begin{proof}
	\begin{align*}
		\II(A : B \mid C) &= \II(A : B,C) - \II(A : C) \\
		&= \II(A : B) + \II(A : C \mid B) - \II(A : C) \\
		&\leq \II(A : B) + \HH(C \mid B) \leq \II(A : B) + \HH(C)
	\end{align*}
	where the first two equalities are by chain rule (\itfacts{chain-rule}), the second inequality is by definition of mutual information and its positivity (\itfacts{info-zero}), and the last one is because conditioning 
	can only reduce the entropy (\itfacts{cond-reduce}).
\end{proof}
\section{Proof of Lemma~\ref{lem:ic-ghd}}\label{app:ic-ghd}
\begin{proof}
	The proof consists of two separate parts. We first prove that there exists a pair $a,b \in [t]$, for which $\GHD$ is still ``hard'' under the distribution $\distU(a,b) := \distU \mid \card{A} = a, \card{B}=b$ (i.e., when we 
	fix the size of the sets $A$ and $B$), and in next 
	part, use this fact to prove the bound for the distribution $\distGHDN$ defined for the same pair of $a,b$ found in the first part (the proof of second part is basically the same as Lemma~\ref{lem:disjointness}). 
	
	\begin{claim}\label{clm:ghd-a-b}
		Let $\delta > 0$ be a sufficiently small constant; there exists a pair of $a,b \in [t]$ such that 
		\[\IC{\GHD}{\distU(a,b)}{\delta} = \Omega(t)\]
		 whereby $\distU(a,b) = \distU \mid \card{A} = a, \card{B} = b$. 
	\end{claim}
	\begin{proof}
		Let $\delta$ be as in Lemma~\ref{lem:ic-ghd-uniform}. Suppose by contradiction that for all $a,b \in [t]$, $\IC{\GHD}{\distU(a,b)}{\delta} = o(t)$, and let $\prot_{a,b}$ be the protocol achieving this bound 
		for a specific choice of $a,b$. We design the following protocol $\prot$ for $\GHD$ on the distribution $\distU$: Given an input $(A,B) \sim \distU$, Alice and Bob first communicate $a=\card{A}$, $b = \card{B}$ 
		to each other and then run $\prot_{a,b}$ on their input and output the same answer as in $\prot_{a,b}$. Since each $\prot_{a,b}$ is computed on the same exact distribution as $\distU(a,b)$ (corresponding
		to the same parameters $a$ and $b$), $\prot$ is a $\delta$-error protocol for $\GHD$ on $\distU$. We now bound the information
		 cost of $\prot$ as follows (in the following, $\Prot$ corresponds to the protocol $\prot$, $\Prot_{a,b}$ 
		corresponds to the protocol $\prot_{a,b}$, and $X_A$ (resp. $X_B$) is a random variable for size of $A$ (resp. $B$)) 
		\begin{align*}
			\ICost{\prot}{\distU} &= \II(A : \Prot \mid B) + \II(B : \Prot \mid A) \\
			&= \II(A : \Prot_{a,b},X_A,X_B \mid B) + \II(B: \Prot_{a,b},X_A,X_B \mid A) \\ 
			&= \II(A: \Prot_{a,b} \mid B,X_A,X_B) +  \II(B : \Prot_{a,b} \mid A,X_A,X_B) + 2\cdot\HH(X_A,X_B)  \tag{by \itfacts{chain-rule}) and \itfacts{info-zero}}\\
			&= \EX_{(a,b)} \Big[ \II(A : \Prot_{a,b} \mid B , X_A = a, X_B = b) + \II(B : \Prot_{a,b} \mid A, X_A = a, X_B = b)\Big] + O(\log{t}) \tag{by~\itfacts{uniform}} \\
			&= \EX_{(a,b)} \Bracket{\II_{\distU(a,b)}(A: \Prot_{a,b} \mid B) + \II_{\distU(a,b)}(B : \Prot_{a,b} \mid A)} + O(\log{t}) \tag{by definition, $\distU(a,b) := \distU \mid X_A = a, X_B = b$} \\
			&= \EX_{(a,b)} \Bracket{\ICost{\prot_{a,b}}{\distU(a,b)}} + O(\log{t}) \tag{by definition of information cost of $\prot_{a,b}$}\\
			&= o(t) + O(\log{t}) = o(t) 
		\end{align*}
		where in the second last inequality we used  assumption that $\ICost{\prot_{a,b}}{\distU(a,b)} = o(t)$ for all $a$ and $b$.

		Consequently, we obtained a $\delta$-error protocol $\prot$ for $\GHD$ on the distribution $\distU$ with information cost of $o(t)$, a contradiction with Lemma~\ref{lem:ic-ghd-uniform}. This means that
		there should exists at least on pair $a,b$ such that $\IC{\GHD}{\distU(a,b)}{\delta}$ is $\Omega(t)$, proving the claim. 
	\end{proof}
	
	Now fix $a$ and $b$ as in Claim~\ref{clm:ghd-a-b} and define $\distGHD$ accordingly. Suppose by contradiction
	that $\ICost{\protGHD}{\distGHD}$ is some $\tau = o(t)$. We can use the previous information odometer argument (i.e., Lemma~\ref{lem:one-sided-two-sided}) 
	to create a protocol $\protGHD'$ that solves $\GHD$ on the distribution $\distU(a,b)$ and has information cost of $\ICost{\protGHD'}{\distU(a,b)} = o(t)$, a
	contradiction with Claim~\ref{clm:ghd-a-b}. We can simply create $\protGHD'$ as follows: run the protocol 
	$\protGHD$ and the information odometer in parallel; whenever the information cost of $\protGHD$ is larger than $c \cdot \tau$ (for a sufficiently large constant $c$), 
	terminate the protocol and output an arbitrary answer, otherwise output the same answer as $\protGHD$. This ensures that  $\ICost{\protGHD'}{\distU(a,b)} = o(t)$. 
	By definition of the distribution $\distGHD$, and the fact that $\ICost{\protGHD}{\distGHD} = \tau$, in the cases that we terminate the protocol $\protGHD$, the answer to $\GHD$ can be arbitrary
	w.p. $1-o(1)$ and hence the new protocol is $(\delta+o(1))$-error protocol for $\GHD$ on $\distU(a,b)$. This can be made formal exactly as in the proof of Lemma~\ref{lem:one-sided-two-sided}. 
	The rest of the proof now follows from Lemma~\ref{lem:one-sided-two-sided} exactly as in Lemma~\ref{lem:disjointness}.
\end{proof}

\end{document}